\documentclass[11pt, letterpaper]{llncs}

\usepackage{microtype}%if unwanted, comment out or use option "draft"
\usepackage[ruled]{algorithm}
\usepackage{algpseudocode}
\usepackage{ioa_code}
\usepackage{amssymb,amsmath,multicol}
\usepackage{algpseudocode}% http://ctan.org/pkg/algorithmicx
\usepackage{lipsum}% http://ctan.org/pkg/lipsum
\usepackage{appendix}
\usepackage{amsfonts}
\usepackage{epsfig}
\usepackage{calc}
\usepackage{color}
\usepackage[all]{xy}
\usepackage{bm}
\usepackage{enumerate}
\usepackage{hyperref}
\usepackage{float}
\usepackage{cite}
\usepackage{graphicx}
\usepackage{fullpage}

\newcommand{\PutW} {{ \it{write-put}}}

\newcommand{\PutData}{{ \it{put-data}}}

\newcommand{\GetData}{{ \it{get-data}}}
\newcommand{\ConfirmData}{\it{confirm-data}}

\newcommand{\PutDataResp}{{ \it{put-data-resp}}}

\newcommand{\GetDataResp}{{ \it{get-data-resp}}}
\newcommand{\ConfirmDataResp}{\it{confirm-data-resp}}
\newcommand{\ConfirmDataTag}{\text{\sc{confirm-data}}}
\newcommand{\QueryList}{\text{\sc{query-list}}}
\newcommand{\RepairList}{\text{\sc{repair-list}}}
\newcommand{\QueryTagData}{\text{\sc{query-tag-data}}}
\newcommand{\RepairTagData}{\text{\sc{repair-tag-data}}}
\newcommand{\CodedElementTag}{\text{\sc{code-elements}}}
\newcommand{\PutDataTag}{\text{\sc{put-data}}}
\newcommand{\QueryTag}{\text{\sc{query-tag}}}
\newcommand{\Coded}{code\act{-}elems}

\newcommand{\GetTag}{{ \it{get-tag}}}
\newcommand{\GetTagResp}{{ \it{get-tag-resp}}}
\newcommand{\InitRep}{{ \it{init-repair}}}
\newcommand{\InitRepResp}{{ \it{init-repair-resp}}}

\newcommand{\RADONL}{{RADON_R}} 
\newcommand{\RADONS}{{RADON_R^{(S)}}} 
\newcommand{\RADONC}{{RADON_C}}

\newcommand{\optag}[1]{{tag(#1)}}

\newcommand{\act}[1]{%
	\relax\ifmmode
	\mathord{\mathcode`\-="702D\sf #1\mathcode`\-="2200}%
	\else
	$\mathord{\mathcode`\-="702D\sf #1\mathcode`\-="2200}$%
	\fi
}

\begin{document}

	\title{RADON: Repairable Atomic Data Object in Networks}	
	\author{Kishori M. Konwar, N. Prakash, Nancy Lynch, Muriel M{\'{e}}dard} 
	\institute{Department of EECS, Massachusetts Institute of Technology \\
		\email{\{kishori, lynch\}@casil.mit.edu, \{prakashn, medard \}@mit.edu}}
	\date{}
	\maketitle

	\begin{abstract}
		
		Erasure codes offer an efficient way to decrease storage and communication costs while implementing atomic memory service in asynchronous distributed storage systems. In this paper, we provide erasure-code-based algorithms having the additional ability to perform background repair of crashed nodes.  A repair operation of a node in the crashed state is triggered externally, and is carried out by the concerned node via message exchanges with other active nodes in the system. Upon completion of repair, the node re-enters active state, and resumes participation in ongoing and future read, write, and repair operations. To guarantee liveness and atomicity simultaneously, existing works assume either the presence of nodes with stable storage, or presence of nodes that never crash during the execution. We demand neither of these; instead we consider a natural, yet practical network stability condition $N1$ that only restricts the number of nodes in the crashed/repair state during broadcast of any message. 
		
		\hspace{0.2in} We present an erasure-code based algorithm $\RADONC$ that is always live, and guarantees atomicity as long as condition $N1$ holds. In situations when the number of concurrent writes is limited, $\RADONC$ has significantly improved storage and communication cost over a replication-based algorithm $\RADONL$, which also works under $N1$. We further show how a slightly stronger network stability condition $N2$ can be used to construct algorithms that never violate atomicity. The guarantee of atomicity comes at the expense of having an additional phase during the read and write operations.

	\end{abstract}

	\section{Introduction} \label{sec:intro}

	We consider the problem of designing algorithms for distributed storage systems (DSSs) that offer consistent access to stored data. Large scale DSSs are widely used by several industries, and also widely studied by academia for a variety of applications ranging from e-commerce to sequencing genomic-data. The most desirable form of consistency is atomicity, which in  simple terms, gives the users of the data service  the impression that the various concurrent read and write operations take place sequentially. Implementations of atomicity on an asynchronous system under message passing framework, in the presence of failures, is often challenging. Traditional implementations\cite{ABD96}, \cite{FL03} use replication of data as the mechanism of fault-tolerance; however they suffer from the problem of having high storage cost, and communication costs for read and write operations.
	
	Erasure codes provide an efficient way to decrease storage and communication cost in atomicity implementations. An $[n, k]$ erasure code splits the value $v$, say of size $1$ unit into $k$ elements, each of size $\frac{1}{k}$ units, creates $n$ {\em coded elements}, and stores one coded element per server. The size of each coded element is also $\frac{1}{k}$ units. A class of erasure codes known as Maximum Distance Separable (MDS) codes have the property that value $v$ can be reconstructed from any $k$ out of these $n$ coded elements. While it is known that usage of erasure codes in asynchronous decentralized storage systems do not offer all the advantages as in synchronous centralized systems~\cite{kedar_bounds}, erasure code based algorithms like in \cite{aguilera}, \cite{DGL08}, \cite{CadambeLMM14}, or \cite{SODA2016} for implementing consistent memory service offer significant storage and communication cost savings over replication based algorithms, in many regimes of operation. For instance CASGC~\cite{CadambeLMM14} improves the costs under the scenario when the number of writes concurrent with a read is known to be limited, whereas SODA~\cite{SODA2016} trades-off write cost in order to optimize storage cost, which is meaningful in systems with infrequent writes.  Both CASGC and SODA are based on MDS codes.

	In this work, we consider the additional important issue of repairing crashed nodes without disrupting the storage service. Failure of storage nodes is a norm rather than an exception in large scale DSSs today, primarily because of the usage of commodity hardware for affordability and scalability reasons. Replication based algorithms in \cite{ABD96}, \cite{FL03} and erasure-code based algorithms in \cite{aguilera}, \cite{CadambeLMM14}, or \cite{SODA2016} do not consider repair of crashed nodes; instead assume that a crashed node remains so for the rest of the execution. Algorithms in \cite{DGL08}, \cite{amnesic} consider background repair of crashed nodes; however they assume either the presence of nodes having stable storage, whose content is unaffected by crashes, or presence of a subset of nodes that never crash during the entire execution. We relax both these assumptions in this work. In our model, any one of the storage nodes can crash; further, we assume that a crashed node loses all its data, both volatile as well as stable storage. A repair operation of a node in the crashed state is triggered externally, and is carried out by the concerned node via message exchanges with other active nodes in the system. Upon completion of repair, the node (with the same id) re-enters active state, and resumes participation in ongoing and future read, write, and repair operations.

	It is natural to expect a restriction on the number of crash and repair operations in relation to the read and write operations; the authors of \cite{amnesic} show an impossibility result in this direction, for guaranteeing liveness and atomicity, simultaneously. We formulate network stability conditions $N1$ and $N2$, which can be used to limit the number of crash and repairs operations overlapping with a client operation. These conditions are algorithm independent, and most likely to be satisfied in any practical storage network. At a high level, the condition $N1$ restricts the  set of servers that can be in the crashed or repair state any time a process (client or server) \emph{pings} all the $n$ servers with corresponding messages. Condition $N2$ is  slightly stronger than $N1$, and restricts the set of servers that can be in the crashed or repair state if the process wants to \emph{ping-pong} a fraction of  the servers. In a ping-pong, it is expected that the servers which receive a message also respond back to the sender of the message.

	\subsection{Summary of Our Contributions}
	
	We first present an impossibility result for an asynchronous DSS allowing background repair of crashed nodes, where there is no restriction on the number of crash and repair operations that occur during a client operation. We show that it is impossible to simultaneously achieve liveness and atomicity in such a system, even if all the crash and repair operations occur sequentially during the execution (i.e., at most one node remains in the crash or repair state at any point during the execution). 
	
	We then consider the problem of erasure-code based algorithm design under the network stability condition $N1$. We present the algorithm in two stages. First we present a replication-based algorithm {$\RADONL$}, which performs background-repair, and guarantees atomicity and liveness of operations under $N1$, if more than ${3/4}^{\text{th}}$ of all servers remain active during any ping operation. The write and read phases are almost identical to those of the ABD algorithm~\cite{ABD96}, except that during a write we expect responses from more than ${3/4}^{\text{th}}$ of all the servers, while in ABD responses are expected only from a majority of servers. A repair operation  in $\RADONL$ is simply a read operation initiated by the concerned server. Thus the algorithm itself is simple; however, the proof of atomicity gets complicated because of the fact that a repair operation can potentially restore the contents of a node to a version that is older than what was present before the crash. We show how the network stability condition can be used to prove atomicity, and this proof is the key takeaway from $\RADONL$ towards constructing the erasure-code based algorithm. 
	
	Our erasure-code based algorithm  $\RADONC$ uses $[n, k]$ MDS codes, and is a natural adaptation of {$\RADONL$} for the usage of codes. A key challenge while using erasure codes is ensuring liveness of read operations, in the presence of concurrent write operations. Various techniques are known in literature to handle this challenge; for instance, \cite{DGL08} assumes synchronous write phases, \cite{CadambeLMM14} limits the number of writes concurrent with a read, while \cite{SODA2016} uses an $O(n^2)$ write protocol to guarantee liveness of reads. In this work, like in \cite{CadambeLMM14}, we make the assumption that the number of write operations concurrent with any read operation is limited by a parameter $\delta$, which is known a priori. However, the usage of the concurrency bound differs from that of the CASGC algorithm in \cite{CadambeLMM14}; for instance, CASGC has three rounds for write operations, while $\RADONC$ uses only two rounds. In $\RADONC$, each server maintains a list of up to $\delta+1$ coded elements, corresponding to the latest $\delta+1$ versions received as a result of the various write operations. In comparison with $\RADONL$ where a writer expects responses from more than ${3/4}^{\text{th}}$ of all servers, a write operation in $\RADONC$ expects responses from more than $\frac{3n+k}{4}$ servers. During a read operation, the client reads the lists from more than $\frac{n+k}{2}$ nodes before decoding the value $v$. Like in $\RADONL$, a repair operation in $\RADONC$ is essentially a read operation by the concerned node; however this time the concerned node creates a list (instead of just one version) by decoding as many possibles versions that it can from the $\left\lceil \frac{n+k}{2} \right \rceil$ responses. Liveness and atomicity of operations are proved under network stability condition $N1$, if more than $\frac{3n+k}{4}$ servers remain active during any ping operation. $\RADONC$ has substantially improved storage and communication costs than $\RADONL$, when the concurrency bound $\delta$ is limited; see Table \ref{table:summ} for a comparison. 
	
	In both $\RADONL$ and $\RADONC$, violation of the network stability condition $N1$ can result in executions that are not atomic, which might not be preferable in certain applications. The choice of consistency over liveness, or vice versa, is the subject matter of a wide range of discussions and perspectives among system designers and software engineers. For example, BigTable,  a DSS by Google,  prefers safety over liveness~\cite{Bigtable}, whereas, Amazon's Dynamo does not compromise liveness but settles for \emph{eventual consistency} ~\cite{DeCandia:2007}.  Our third algorithm {$\RADONS$}, which is replication-based, is designed to guarantee atomicity during every execution. Liveness is guaranteed under the slightly  more stringent condition of $N2$, with more than ${3/4}^{\text{th}}$ of all servers remaining active during any ping-pong operation. The guarantee of atomicity of every execution also needs extra phases for read and write operations, when compared to {$\RADONL$}. The design of an erasure-coded version of {$\RADONS$} that never violates atomicity, is an interesting direction that we leave out for future work. 
	
	\begin{table}[ht]
		\centering
		\begin{tabular}{c c c c c c}
			\hline
			Algorithm & Write Cost & Read Cost   & Storage Cost    & Safe under   & Live under \\
			\hline
			{$\RADONL$} & $n$ & $2n$  & $n$ & $N1$ & $N1$  \\
			{$\RADONC$} & $\frac{n}{k}$ & $(\delta+2)\frac{n}{k}$  &  $( \delta+1)\frac{n}{k}$ & $N1$ & $N1$ \\
			{$\RADONS$} &$n$ & $2n$ &  $n$ & $always$ & $N2$  \\
			%	{$\RADONC$} & $\frac{n}{k}$ & $(\delta+2)\frac{n}{k}$  &  $( \delta+1)\frac{n}{k}$ & $N1$ & $N1$ \\
			\hline
			\vspace{0.1in}
		\end{tabular}
		\caption{Performance comparison of ${\RADONL}$, ${\RADONC}$ and ${\RADONS}$, where $n$ is the number of servers, and $\delta$ is the maximum number of writes concurrent with a read or a repair operation. See Section \ref{sec:costs} for a justification of the costs.}  \label{table:summ}
		\vspace{-0.5in}	
	\end{table}
		
\subsection{Other Related Work}

\emph{Dynamic Reconfiguration: } Our setting is closely related to the problem of implementing a consistent memory object in a dynamic setting, where nodes are allowed to voluntarily leave and join the network. The problem involves dynamic reconfiguration of the set of nodes that take part in client operations, which is often implemented via a \textit{reconfig} operation that is initiated by any of the participating processes, including the clients.  Any node that wants to leave/join the network makes an announcement, via a \textit{leave/join} operation, before doing so. The problem is extensively studied in the field of distributed algorithms~\cite{LS02}, \cite{AKMS11}, \cite{Alex_Idit}, \cite{BBKR09}, \cite{AttiyaCEKW15}; review and tutorial articles appear in \cite{Aug_tutorial},  \cite{Spi_tutorial}, \cite{Mus_tutorial}. 

In our context, the problem of node repair could in fact be thought of as one of dynamic reconfiguration, wherein an involuntary crash is simulated by a voluntary leave operation without an explicit announcement.  In this case, a new node joins as a replacement node via the \textit{join} operation, which can be considered as the analogue of a \textit{repair} operation.  In the setting of dynamic reconfiguration, every node has a distinct identity; thus the replacement node joins the network with a new identity that is different from the identity of the crashed node~\cite{Aug_tutorial}. This demands a reconfiguration of the set of participating nodes after every repair. Such reconfigurations get in the way of client operations, and add to the latency of read and write operations~\cite{Mus_tutorial}, in practical implementations. Clearly, a repair operation as considered in this work does not demand any reconfiguration, since a repaired node has the same identity as the crashed node. Also, the current work shows that modeling repair via a static system, permits design of algorithms where clients remain oblivious to the presence of repair operations. Furthermore, addressing storage and communication costs is not the focus of the works in dynamic reconfigurations; specifically, it is not known as to how erasure codes can be advantageously used in dynamic settings. Our $\RADONC$ algorithm shows that when repair is carried out under a static model, it is indeed possible to advantageously use erasure to reduce costs, when the number of concurrent writes are limited. 

We make additional comparisons between our model and results to those found in works on dynamic reconfiguration.
Several impossibility results exist in the context of implementing a dynamic atomic register and simultaneously guaranteeing liveness; the authors in \cite{Alex_Idit} argue impossibility if there are infinitely many reconfigurations during an execution, while the authors in \cite{BBKR09} argue an impossibility when there is no upper bound on message delay. We see, not surprisingly, that even in the problem of repair, we need to suitably limit the number of crash and repair operations that occur in an execution, even if all crash and repairs are sequentially ordered. In \cite{AttiyaCEKW15}, the authors implement a dynamic atomic register  under a model that has an (unknown) upper bound $D$ on any point-to-point message delay, and where the number of reconfigurations in any $D$ units of time is limited. Our network condition $N1$ is similar, except that $1)$ we limit the number of crash and repairs during any broadcast messaging, instead of point-to-point messaging, and $2)$ we do not assume any bound on the message delay. 
In practice, limiting number of repairs during broadcast instead of every point-to-point messaging offers resiliency against \textit{straggler} nodes, which refer to the nodes having the worst delays among all nodes. We would also like to note that the algorithm in \cite{AttiyaCEKW15} does not guarantee atomicity, if the number of reconfigurations in $D$ units of time is higher than a set number. 
This appears similar to $\RADONL$, where atomicity is not guaranteed if we do not satisfy stability condition $N1$. While we show how the slightly tighter model $N2$ can be used to always guarantee atomicity, it is an interesting question as to whether the model $N2$ can be adopted in the work of \cite{AttiyaCEKW15} so as to always guarantee atomicity.

\emph{Repair-Efficient Erasure Codes for Distributed Storage:} Recently, a large class of new erasure/network codes for storage have been proposed (see ~\cite{dimakis2011survey} for a survey), and also tested in networks ~\cite{HuaSimXu_etal_azure}, \cite{sathiamoorthy}, \cite{rashmi_fast15},  where the focus is efficient storage of immutable data, such as, archival data.  These new codes are specifically designed to optimize performance metrics like repair-bandwidth and repair-time (of failed servers), and offer significant performance gains when compared to the traditional Reed-Solomon MDS codes~\cite{rscodes}. It needs to be explored if these codes can be used in conjunction with the {$\RADONC$} algorithm, to further improve the performance costs.

\emph{Other Works on using Erasure Codes:} Applications of erasure codes to Byzantine fault tolerant DSSs are discussed in \cite{cachin}, \cite{dobre}, \cite{hendricks}.  In \cite{kedar_bounds}, the authors consider algorithms that use erasure codes for emulating {\em regular} registers. Regularity~\cite{regular_lamport}, \cite{shao} is a weaker consistency notion than atomicity.

The rest of the document is organized as follows. Our system model appears in Section \ref{sec:models}. The impossibility result, and the network stability conditions appear in Section \ref{sec:net_stab}. The three algorithms appear in Sections \ref{sec:radonl}, \ref{sec:radonc} and \ref{sec:radons}, respectively. In Section \ref{sec:costs}, we discuss the storage and communication costs of the algorithms. Section ~\ref{sec:conclusion} concludes the paper. Proofs of various claims appear in the Appendix.
	
	\section{Models and definitions} \label{sec:models}

	\emph{{\bf Processes and Asynchrony}: } We consider a distributed system consisting of \emph{asynchronous} processes, each with a unique identifier (ID),  {of three types: a set of \emph{readers}, ${\mathcal R}$;  a set of \emph{writers}, ${\mathcal W}$; and a set of $n$ \emph{servers}, ${\mathcal S}$.} The readers and writers are together referred to as clients.  
	%We denote the sets of IDs of the readers, writers and servers as ${\mathcal R}$, ${\mathcal W}$ and  ${\mathcal S}$, respectively. 
	The set  ${\mathcal R}\cup{\mathcal W}\cup{\mathcal S}$ forms a  totally ordered set under some defined relation ($>$). The reader and writer processes initiate {\em read} and {\em write}  operations respectively,  and communicate with the servers using messages. 
	{A reader or writer can invoke a new operation only after all previous operations invoked by it has completed. The property is referred to as the {\em well-formedness} property of an execution.} 
	We assume that every client/server is connected to every other server via a reliable communication link; thus as long as the destination process is non-faulty, any message sent on  the link  eventually reaches the destination process.  
	
	\vspace{0.1in}
	
	\noindent \emph{{\bf Crash and Recovery}:} A client may fail at any point during the execution. At any point during the execution, a server can be in one (and only one) of the following three states: \emph{active}, \emph{crashed} or \emph{repair}. A crash event triggers a server to enter the \emph{crashed} state from an \emph{active} state. The server remains in the \emph{crashed} state for an arbitrary amount of time, but eventually is triggered by a repair event to enter the \emph{repair} state. Crash and repair events are assumed to be externally triggered.  A server in the \emph{repair} state can experience another crash event, and go back to the \emph{crashed} state.  A server in the \emph{crashed} state does not perform  any local computation. The server also does not send or receive messages in the \emph{crashed} state, i.e., any message reaching the server in a \emph{crashed} state is lost. A server which enters the \emph{repair} state has all its local state variables set to default values, i.e., a crash event causes the server to lose all its state variables. A server in the \emph{repair} state can perform computations like in the \emph{active} state. 
	
	\vspace{0.1in}
	
	\noindent \emph{{\bf Atomicity and Liveness}:} We aim to implement only one atomic read/write memory object, say $x$, under the MWMR setting on a set of  servers, because any shared atomic memory can be emulated by composing individual atomic objects.
	The object value $v$ comes from some set $V$; initially 	$v$ is set to a distinguished value $v_0$ ($\in V$). Reader $r$ requests a read operation on  object $x$. Similarly, a write operation is requested by a writer $w$. Each operation at a non-faulty client begins with an \emph{invocation step} and terminates with a  \emph{response step}. An operation  is \emph{incomplete} when its invocation step does not have the associated response step; otherwise it is  \emph{complete}. 
	
	By \emph{liveness of a read or a write operation}, we mean that during any well-formed execution, any read or write operation respectively initiated by a non-faulty reader or writer completes,  despite the crash failure of any other client. By \emph{liveness of repair} associated with a crashed server, we mean that the server which enters a crashed state eventually re-enters the active state, unless it experiences a crash event during every repair operation that the server attempts. The liveness of repair holds despite the crash failure of any other client. 
	
	\vspace{0.1in}
	
	\noindent \emph{{\bf Background on Erasure coding}:} In $\RADONC$, we use an $[n, k]$  linear MDS code ~\cite{verapless_book} over a finite field $\mathbb{F}_q$ to encode and store the value $v$ among the $n$ servers. An $[n, k]$ MDS code has the property that any $k$ out of the $n$ coded elements can be used to recover (decode) the value $v$. For encoding, $v$ is divided\footnote{In practice $v$ is a file, which is divided into many stripes based on the choice of the code, various stripes are individually encoded and stacked against each other. We omit details of representability of $v$ by a sequence of symbols of $\mathbb{F}_q$, and the mechanism of data striping, since these are fairly standard in the coding theory literature.} into $k$ elements $v_1, v_2, \ldots v_k$ with each element having  size $\frac{1}{k}$ (assuming size of $v$ is $1$). The encoder takes the $k$ elements as input and produces $n$ coded elements $c_1, c_2, \ldots, c_n$ as output, i.e., $[c_1, \ldots, c_n] = \Phi([v_1, \ldots, v_k])$, where $\Phi$ denotes the encoder. For ease of notation, we simply write $\Phi(v)$ to mean  $[c_1, \ldots, c_n]$. The vector $[c_1, \ldots, c_n]$ is  referred to as the codeword corresponding to the value $v$. Each coded element $c_i$ also has  size $\frac{1}{k}$. In our scheme we store one coded element per server. We use $\Phi_i$ to denote the projection of $\Phi$ on to the $i^{\text{th}}$ output component, i.e., $c_i = \Phi_i(v)$. Without loss of generality, we associate the coded element $c_i$ with server $i$, $1 \leq i \leq n$.
	
	\vspace{0.1in}
	
	\noindent \emph{{\bf Storage and Communication Cost}:}  We define the total storage cost as the size of the data stored across all servers, at any point during the execution of  the algorithm.  The communication cost associated with a read or write operation is the size of the total data that gets transmitted in the messages sent as  part of the operation. We assume that metadata, such as version number, process ID, etc. used by various operations is of negligible size, and is hence ignored in the calculation of storage and communication cost. Further, we normalize both the costs with respect to the size of the value $v$; in other words, we compute the costs under the assumption that $v$ has size $1$ unit.

	\section {Network Stability Conditions} \label{sec:net_stab}

	\subsection{An Impossibility Result}
%	\vspace{-0.1in}
	The crash and recovery model described  in Section \ref{sec:models} does not impose any restriction on the \emph{rate of crash events, and repair operations} that happen in the system. In other words, the model described above does not limit in any manner the number of crash events/repair operations, which can overlap with any a client operation. In \cite{amnesic}, the authors showed that without such restrictions, it is impossible to implement a shared atomic memory service, which guarantees liveness of operations. Below, we state an impossibility result which holds even if there is at most one server in the crashed/repair state at any point during the execution. We then introduce network stability conditions that enable us impose restrictions on the number of crash/repair events that overlap with any  operation.  

	\begin{theorem} \label{thm:imp_weak}
{	It is impossible to implement an atomic memory service that guarantees liveness of reads and writes,  under the system model described in Section \ref{sec:models}, even if $1$) there is at most one server in the crashed/repair state at any point during the execution, and $2$) every repair operation completes, and takes the repaired server back to the active state.}
	\end{theorem}

	\subsection{Network Stability Conditions $N1$ and $N2$ }
	
	We begin with the notions of a group-send operation, and effective consumption of a message. 

	\vspace{0.1in}
	\noindent \emph{{\bf group-send operation}:} The group-send operation is used to abstract the operation of a process sending a list of  $n$ messages $\{m_1, \cdots, m_n\}$ to the set of all $n$ servers $\{s_1, \ldots, s_n\} = \mathcal{S}$, where message $m_i$ is send to server $s_i, 1 \leq  i \leq n$. Note that this is a mere abstraction of the process sending out $n$ point-to-point  messages sequentially  to  $n$ servers, without interleaving the ``send" operations with any significant local computations or waiting for any external inputs. The operation is no more powerful then sending $n$ consecutive messages.  The operation is written as $group\act{-}send([m_1, m_2, \cdots, m_n])$. In the event  $m_i = m, \forall i$, we simply write $group\act{-}send(m)$. Our model allows the sender to fail while executing the $group\act{-}send$ operation, in which case only a subset of the $n$ servers receive their corresponding messages.
	
	\vspace{0.1in}
	\noindent \emph{{\bf Effective Consumption}:} We say a process effectively consumes a message $m$, if it receives  $m$, and executes all steps of the algorithm that depend only on the local state of the process, and the message $m$; in other words, the process executes all the steps that do not require any further external messages.
	
	\begin{definition}[Network Stability Conditions]
		Consider a process $p$ executing a $group\act{-}send$ $([m_1, m_2, \cdots, m_n])$ operation, and 
		consider the following statements:
		
		$(a)$   $(i)$ There exists a subset $\mathcal{S}_{\alpha} \subseteq {\mathcal S}$ of 
		$|\mathcal{S}_{\alpha}| = \left\lceil\alpha n\right\rceil$ servers, $0 < \alpha < 1$, all of which effectively consume their respective messages from the group-send operation, and $(ii)$ all the servers in $\mathcal{S}_{\alpha}$ remain in the active state during the interval $[T_1 \ T_2]$, where $T_1$ denotes the point of time of invocation of the  $group\act{-}send$ operation, and  $T_2$ denotes the earliest point of time in the execution at which all of the servers in $\mathcal{S}_{\alpha}$ complete the effective consumption of their respective messages.
		
		$(b)$ Further, if effective consumption of the message $m_i$ by server $s_i$ involves sending a response back to the process $p$, for all $s_i \in \mathcal{S}_{\alpha}$, then all servers in $\mathcal{S}_{\alpha}$ remain in the active state during the interval $[T_1 \ T_3]$, where $T_3$ denotes the earliest point of time in the execution at which the process $p$ completes effective consumption of the responses from the all the servers in $\mathcal{S}_{\alpha}$. 
		
		If the network satisfies Statement $(a)$ for every execution of a group-send operation by any process, we say that it satisfies network stability condition $N1$ with parameter $\alpha$. If the network satisfies Statements $(a)$ and $(b)$ for every execution of a group-send operation by any process, we say that it satisfies network stability condition $N2$ with parameter $\alpha$. 
	\end{definition}
	
	Clearly,  $N2$ implies $N1$. Note that the set $\mathcal{S}_{\alpha}$ which needs to satisfy the conditions need not be the same for various invocations of group-send operations by either the same or distinct processes. Also, note that in condition $N2$, the process $p$ might crash before completing the effective consumption of the responses from the servers in $\mathcal{S}_{\alpha}$. In this case we only expect Statement $(a)$ to be satisfied, and not Statement $(b)$. Furthermore, in both $N1$ and $N2$, we do not expect any of these statements to be true, if process $p$ crashes after partial execution of the group-send operation.
	
	\section{The $\RADONL$ Algorithm}\label{sec:radonl}
%	\vspace{-0.15in}	
	In this section,  we  present the $\RADONL$ algorithm, and prove its liveness and atomicity properties for networks that satisfy the network condition $N1$ with $\alpha > \frac{3}{4}$.  We begin with some useful notation. Tags are used for version control of the  object values. A tag $t$ is defined as a pair $(z, w)$, where $z \in \mathbb{N}$ and $w \in \mathcal{W}$ denotes the ID of a writer. We use $\mathcal{T}$ to denote the set of all the possible tags. For any two tags $t_1, t_2 \in \mathcal{T}$, we say  $t_2 > t_1$ if $(i)$ $t_2.z > t_1.z$ or $(ii)$ $t_2.z = t_1.z$ and $t_2.w > t_1.w$. Note that $(\mathcal{T}, >)$ is  a totally ordered set. 
	
	The protocols for writer, reader, and servers are shown in Fig.~\ref{fig:radonl-server}. Each server stores two state variables $(i)$ $(t_{loc}, v_{loc})$ -  a tag and  value pair, initially set to $(t_0, v_0)$,  $(ii)$ $status$ - a variable that can be in either \emph{active} or \emph{repair} state. 
	
	\begin{algorithm*}[!ht]
		\begin{algorithmic}
			\begin{multicols}{2}{\footnotesize
					\Part{write(v)} { }\EndPart
					\Part{\underline{\GetTag}} {
						\State  $group\act{-}send(\QueryTag)$
						\State  Await  responses from majority
						\State  Select the max tag  $t^*$
					}\EndPart
					
					\Part{\underline{\PutData}} {
						\State $t_w = (t^{*}.z + 1, w)$  
						\State $group\act{-}send((\PutDataTag, (t_w, v)))$
						\State Terminate after $\left\lceil \frac{3n + 1}{4} \right\rceil$ acks.
					}	\EndPart
					\Statex
				
					\Part{read} { }\EndPart
					\Part{\underline{\GetData}} {
						\State  $group\act{-}send(\QueryTagData)$
						\State  Await responses from  majority
						\State  Select $(t_r, v_r)$, with  max tag.
					}\EndPart
					\Statex	
					\Part{\underline{{\PutData} }}{
						\State $group\act{-}send((\PutDataTag, (t_r, v_r)))$
						\State Wait for $\left\lceil \frac{3n + 1}{4}\right\rceil$ acks 
						\State Return $v_r$
					}\EndPart
					
					\Statex
					\Part{Server $s \in \mathcal{S}$}\EndPart
					\Part {\underline{$State~Variables$}}{ 
						
						\Statex $(t_{loc}, v_{loc}) \in \mathcal{T} \times {\mathcal V}$, initially   $(t_0, v_0)$
						\Statex $status \in \{active, repair\}$, initially $active$
						
					}\EndPart
					%						\Statex
					\Part {\underline{\GetTagResp, recv $\QueryTag$~from writer $w$}} {
						\If{ $status = active$ }
						\State Send  $t_{loc}$ to $w$
						\EndIf
					}\EndPart
					%						\Statex
					\Part {\underline{\GetDataResp, recv $\QueryTagData$~from reader $r$}} {
						\If{ $status = active$ }
						\State Send  $(t_{loc}, v_{loc})$ to $r$
						\EndIf
					}\EndPart
					\Statex
					\Part{ \underline{\PutDataResp,~recv $\PutDataTag, (t, v)$ from client $c$ }}{
						\If{$status = active$} 
						\If{ $t > t_{loc}$ } 
						\State  $(t_{loc}, v_{loc}) \leftarrow (t, v)$
						\EndIf 
						\State  Send ack to $c$.
						\EndIf
						
					}\EndPart
											\Statex
					\Part{ \underline{\InitRep} }{
						\State $status \leftarrow repair$
						\State $(t_{loc}, v_{loc}) \leftarrow (t_0, v_0)$
						\State $group\act{-}send(\RepairTagData)$
						\State Await responses from  majority
						\State Select $(t_{rep}, v_{rep})$, for max tag 
						\State $(t_{loc}, v_{loc}) \leftarrow (t_{rep}, v_{rep})$
						\State $status \leftarrow active$ 
					}\EndPart
					%						\Statex
					\Part{ \underline{\InitRepResp, recv $\RepairTagData$ from $s'$}}{
						\If{$status = active$} 
						\State  Send $(t_{loc}, v_{loc})$ to $s'$
						\EndIf
					}\EndPart
				}\end{multicols}
			\end{algorithmic}	
			\caption{The protocols for writer, reader, and any  server $s \in {\mathcal S}$ in  $\RADONL$.}\label{fig:radonl-server}
		\end{algorithm*}

		The write and read operations are very similar to those in the ABD algorithm~\cite{ABD96}, and each consists of two phases. In the first phase,\GetTag, of a write operation $\pi$, the writer queries all servers for local tags, awaits responses from a majority of servers, and selects the maximum tag $t^*$ from among the responses.  Next, the writer executes the \PutData~ phase, during which a new tag $t_w = tag(\pi)$ is created by incrementing the integer part of $t^*$, and by incorporating the writer's own ID. The writer then sends pair $(t_w, v)$ to all servers, and awaits acknowledgments (acks) from $\left\lceil \frac{3n+1}{4}\right\rceil$ servers before completing the operation. The two phases are identical to those of the ABD algorithm~\cite{ABD96}, except for the fact that during the second phase, ABD expects acks from only a majority of servers, whereas  here we need from  $\left\lceil \frac{3n+1}{4}\right\rceil$ servers. During a read operation $\rho$, the reader in the {\GetData} phase queries all the servers in $S$ for the respective local tag and value pairs. Onces it receives responses from a majority of servers in $S$, it picks the pair with the highest  
		tag, which we designate as $t_r = tag(\pi)$. In the subsequent \PutData~ phase, the reader
		writes back  the tag $t_r$ and the corresponding value $v_r$ to all servers, and terminates after receiving acknowledgments from $\left\lceil\frac{3n+1}{4}\right\rceil$ servers. Once again, we remark that both phases in the read are identical to those of the ABD algorithm, except for the difference in the number of the servers from which acks are expected in the second write-back phase. Note that, during both the write and operations, a server responds to an incoming message only if it is in the active state. 
		
		A repair operation is initiated via the action \InitRep, by an external trigger, at a server which is in the crashed state. Note that we do not explicitly define a \emph{crashed} state since a crash is not a part of the algorithm. We assume that as soon as the repair operation starts, the variable \textit{status} is set to the \emph{repair} state, and also the local (tag, value) pair is set to the default sate $(t_0, v_0)$.  The repair operation is essentially the first phase of the read operation, during which  the server queries all the servers for the respective local tag and value pairs, and stores the tag and value pair corresponding to the highest tag after receiving responses from a majority of servers. Finally, the repair operation is terminated setting  variable \textit{status} to  \emph{active} state. A server in $S$ responds to a request generated from \InitRep~phase only if it is in the active state.

		\subsection{Analysis of $\RADONL$} \label{sec:analysis_randonl}	                             

		Liveness of read, write and repair operations in $\RADONL$ follows immediately if we assume condition  $N1$ with $\alpha > \frac{3}{4}$. This is because liveness of any operation depends on sufficient number of responses from the servers during the various phases of the operation. From Fig. \ref{fig:radonl-server}, we know that the maximum number of responses that is expected in any phase is $\left\lceil\frac{3n+1}{4}\right\rceil$, which is guaranteed under $N1$ with $\alpha > \frac{3}{4}$. 
		
		The tricky part is to prove atomicity of reads and writes. { The proof  is based on Lemma $13.16$ of \cite{Lynch1996}, a restatement of which can be found in ~\cite{RADON:arxiv}.} Consider two completed write operations $\pi_1$ and $\pi_2$, such that, $\pi_2$ starts after the completion of $\pi_1$. For any completed write operation $\pi$, we define $tag(\pi) = t_w$, where $t_w$ is the tag which the writer uses in the {\PutData} phase. In this case, one of the requirements the algorithm needs to satisfy to ensure atomicity is $tag(\pi_2) > tag(\pi_1)$. While this fact is straightforward to prove for an algorithm like ABD, which does not have background repair, in $\RADONL$, we need to consider the effect of those repair operations that overlap with $\pi_1$, and also those that occur in between $\pi_1$ and $\pi_2$. The point to note is that such repair operations  can potentially restore the contents of the repaired node such that the restored tag is less than $tag(\pi_1)$. We then need to show the absence of propagation of older tags (older than $tag(\pi_1)$) into a majority of nodes, due to a sequence of repairs which happen before $\pi_2$ decides its tag. We do this via the following two observations: $1)$ In Lemma \ref{lem:property}, we show that any successful repair operation, which begins after a point of time $T$, always restores value to one, which corresponds to a tag which is at least as high as the minimum of the tags stored in any majority of active servers at time $T$. This fact is in turn used to prove a similar property for reads and writes, as well. $2)$ We next show (as part of proof of Theorem \ref{thm:atomicity}), under the assumption of $N1$ with $\alpha > 3/4$, the existence of a point of  time $T$  before the completion of $\pi_1$ such that a majority of nodes are active at $T$, and all of whose tags are at least as high as $tag(\pi_1)$. The two steps are together used to prove that $tag(\pi_2) > tag(\pi_1)$. A similar sequence of steps are used to show atomicity properties of read operations, as well. 
		
		For a completed read operation $\pi$, $tag(\pi) = t_r$, where $t_r$ is the tag corresponding to the value $v_r$ returned by the reader. For a completed repair $\pi$, $tag(\pi) = t_{rep}$, where $t_{rep}$ is the tag corresponding to the value restored during the repair operation. 
		
		\begin{lemma}  \label{lem:property}
			Let $\beta$ denote a well-formed execution of $\RADONL$. Suppose $T$ denotes a point of time in $\beta$ such that there exists a majority of servers $\mathcal{S}_m$, $\mathcal{S}_m \subset \mathcal{S}$ all of which are in the active state at  time $T$. Also, let $t_s$ denote the value of the local tag at server $s \in \mathcal{S}_m$, at time $T$. Then, if $\pi$ denotes any completed repair or read operation that is initiated after time $T$, we have $tag(\pi) \geq \min_{s\in \mathcal{S}_m} t_s$. Also, if $\pi$ denotes any completed write  operation that is initiated after time $T$, we have $tag(\pi) > \min_{s\in \mathcal{S}_m} t_s$.
		\end{lemma}

		\begin{theorem}[Liveness]  \label{thm:liveness}
			Let $\gamma$ denote a well-formed execution of $\RADONL$, under the condition $N1$  with $\alpha > \frac{3}{4}$ . Then every operation initiated by a non-faulty client completes. 
		\end{theorem}

		\begin{theorem}[Atomcity]  \label{thm:atomicity}
			Every execution of the $\RADONL$ algorithm operating under the $N1$ network stability condition with $\alpha > \frac{3}{4}$, is atomic.
		\end{theorem}
		
		We note that, though Lemma \ref{lem:property} gives a result about completed operations, condition $N1$ is not a prerequisite for the result in Lemma \ref{lem:property}. In other words, the result in Lemma \ref{lem:property} holds for any completed operation, even if condition $N1$ is violated. As we will see, this is an important fact that we will use to establish atomicity of $\RADONS$ for any execution.

\section{Algorithm ${\RADONC}$ } \label{sec:radonc}
			\vspace{-0.1in}
			In this section, we present the erasure-code based $\RADONC$~algorithm for implementing atomic memory service,  and performing repair of crashed nodes. The algorithm uses $[n, k]$ MDS codes for storage. Liveness and atomicity are guaranteed under the following assumptions: $1)$ the $N1$ network stability condition with $\alpha \geq \frac{3n+k}{4n}$,  $2)$ the number of write operations concurrent with a read or repair operation is at most $\delta$. The precise definition of concurrency depends on the algorithm itself, and appears later in this section. The $\RADONC$~algorithm has significantly reduced storage and communication cost requirements than $\RADONL$,  when $\delta$ is limited.
			
			\begin{algorithm*}[!ht]
				\begin{algorithmic}
					\begin{multicols}{2}{\footnotesize
							\Part{write($v$)}\EndPart
							\Part{\underline{\GetTag}} {
								\State  $group\act{-}send(\QueryTag)$
								\State  Await responses from majority
								\State  Select the max tag  $t^*$
							}\EndPart
							\Statex
							\Part{\underline{\PutData}} {
								\State $t_w = (t^{*}.z + 1, w)$.  
								\State $\Coded = [(t_w, c_1), \ldots, (t_w, c_n)]$, $c_i = \Phi_i(v)$
								\State $group\act{-}send(\CodedElementTag, \Coded)$
								\State Terminate after $\left\lceil \frac{3n + k}{4}\right\rceil$ acks
							}	\EndPart
						
							\Statex
							\Part{read}\EndPart
							\Part{\underline{\GetData}} {
								\State  $group\act{-}send(\QueryList)$
								\State  Wait for $\left\lceil \frac{n+k}{2}\right\rceil$ $Lists$ 
								\State  Select the max tag, $t_r$, whose corresponding value, $v_r$, is decodable using the $Lists$.
							}\EndPart	
							\Statex
							\Part{\underline{\PutData}} {
								\State $\Coded = [(t_r, c_1), \ldots, (t_r, c_n)]$, $c_i = \Phi_i(v_r)$
								\State $group\act{-}send(\CodedElementTag, \Coded)$
								\State Wait for $\left\lceil \frac{3n + k}{4}\right\rceil$ acks
								\State Return $v_r$
							}	\EndPart

							\Statex
							\Part{Server $s_i \in \mathcal{S}$}\EndPart
							\Part {\underline{$State~Variables$}}{ 										
								%\Statex $(t_{loc}, v_{loc}) \in \mathcal{T} \times {\mathcal V}$, initially   $(t_0, v_0)$
								\Statex $status \in \{active, repair\}$, initially $active$
								\Statex $List \subseteq  \mathcal{T} \times \mathcal{C}_s$, initially   $\{(t_0, \Phi_i(v_0))\}$
								
							}\EndPart
							\Statex
							\Part {\underline{\GetTagResp,recv $\QueryTag$ from writer $w$}} {
								\If{ $status = active$ }
								\State $t^* = \max_{(t,c) \in List}t$
								\State Send $t^*$ to $w$
								\Statex \EndIf
							}\EndPart
							%										\Statex
							\Part {\underline{\GetDataResp, recv $\QueryList$ from reader $r$}} {
								\If{ $status = active$ }
								\State Send  $List$ to $r$
								\EndIf
							}\EndPart
							%										\Statex
							\Part{ \underline{\PutDataResp, recv $\CodedElementTag, (t, c_i)$ from $p$ }}{
								\If{$status = active$}
								\State $List \leftarrow List \cup \{ (t, c_i)  \}$ 
								\If{ $|List| > \delta + 1$ } 
								\State  Retain the (tag, coded-element) pairs for the $\delta +1 $ highest tags in $List$, and delete the rest.
								\EndIf 
								\State  Send ack to $p$.
								\EndIf
								
							}\EndPart

							\Part{ \underline{\InitRep} }{
								\State $status \leftarrow repair$
								\State $group\act{-}send(\RepairList)$
								\State Wait for $\left\lceil \frac{n+k}{2}\right\rceil$ $Lists$ 
								\State Find  (tag, value) pairs decodable from  $Lists$.
								\State  Restore local $List$ via re-encoding and retaining the (tag, coded-element) pairs corresponding to at most $\delta +1 $ highest tags, from the above pairs
								\State $status \leftarrow active$ 
							}\EndPart

							\Part{ \underline{\InitRepResp, recv $\RepairList$ from server $s'$}}{
								\If{$status=active$} 
								\State  Send $List$ to $s'$
								\EndIf
							}\EndPart
						}\end{multicols}
					\end{algorithmic}	
					\caption{The protocols for write, reader, and any server $s_i \in {\mathcal S}$ in  
						$\RADONC$.}\label{fig:radonc-server}
				\end{algorithm*}
				
				The algorithm (see Fig.~\ref{fig:radonc-server}) is a natural generalization of the $\RADONL$ algorithm accounting for the fact that we use MDS codes. The write operation has two phases, where the first phase finds the maximum tag in the system based on majority responses. During the second phase, the writer computes the coded elements for each of the $n$ servers and uses the group-send operation to disperse them. The $group\act{-}send$ operation here uses a vector of length $n$, where the $i^{\text{th}}$ element denotes the message for the $i^{\text{th}}$ server, $1 \leq i \leq n$.  Each server keeps a $List$ of up to $(\delta + 1)$  (tag, coded-element) pairs. Every time a (tag, coded-element) message arrives from a writer, the pair gets added to the $List$, which is then pruned to at most $(\delta+1)$ pairs, corresponding to the highest tags. The writer terminates after getting acks from $\left\lceil \frac{3n+k}{4} \right\rceil$ servers.
				
				During a read operation, the reader queries all servers for their entire local $Lists$, and awaits responses from $\left\lceil \frac{n+k}{2} \right\rceil$ servers. Once the reader receives $Lists$ from $\left\lceil \frac{n+k}{2} \right\rceil$ servers,  it selects the highest tag $t_r$ whose corresponding value $v_r$ can be decoded using the using the coded elements in the lists. The read operation completes following a write-back of $(t_r, v_r)$ using the {\PutData} phase.

				The repair operation is very similar to the first phase of the read operation, during which a server collects lists from $\left\lceil \frac{n+k}{2} \right\rceil$ servers. But this time, the server figures out the set of all the possible tags that can be decoded from among the $Lists$, and prunes the set to the highest $(\delta + 1)$ tags. The repaired $List$ then consists of (tag, coded-element) pairs corresponding these (at most) $(\delta + 1)$ tags. Assuming repair of server $i$, the creation of a coded-element corresponding to a value $v$ involves first decoding the value $v$, and then computing $\Phi_i(v)$ (referred to as re-encoding in Fig. \ref{fig:radonc-server}).

				\subsection{Analysis of $\RADONC$}
				
				Throughout this section, we assume network stability condition $N1$ with $\alpha \geq \frac{3n+k}{4n}$. Tags for completed read and write operations are defined in the same manner as we did for $\RADONL$; we avoid repeating them here.  We first discuss liveness properties of $\RADONC$. Let us first consider liveness of repair operations. Towards this, note from the algorithm in Fig. \ref{fig:radonc-server} that a repair operation never gets stuck even if it does not find any set of $k$  $Lists$ among the responses, all of which have a common tag. In such a case, the algorithm allows the possibility that the repaired $List$ is simply empty, at the point of execution when the server re-enters the active state. In other words, liveness of a repair operation is trivially proved, i.e., a server in a repair state always eventually reenters the active state, as long as it does not experience a crash during the repair operation. The triviality of liveness of repair operations, observed above, does not extend to read operations. For a read operation to complete the {\GetData} phase, it must be able to find a set of $k$ $Lists$ among the responses all of which contain coded-elements corresponding to a common tag; otherwise a read operation gets stuck. The discussion above motivates the following definitions of valid read and valid repair operations. 
				
				\begin{definition}[Valid Read and Repair Operations]
				A read  operation will be called as a valid read if the associated reader remains alive at least until the reception of the $\left\lceil \frac{n+k}{2} \right\rceil$ responses during the {\GetData} phase. Similarly, a repair operation will be called a valid repair if the associated server does not experience a further crash event during the repair operation.
				\end{definition}

				\begin{definition}[Writes Concurrent with a Valid Read (Repair)] \label{defn:concurrent}
					Consider a valid read (repair) operation $\pi$. Let $T_1$ denote the point of initiation of $\pi$. For a valid read, let $T_2$ denote the earliest point of time during the execution when the associated reader receives all the 
					$\left\lceil \frac{n+k}{2} \right\rceil$ responses. For a valid repair,  let $T_2$ denote the point of time during the execution when the repair completes, and takes the associated server back to the active state. Consider the set $\Sigma = \{ \sigma: \sigma$ is any write operation that completes before $\pi \text{ is initiated} \}$, and let $\sigma^* = \arg\max_{\sigma \in \Sigma}tag(\sigma)$. Next, consider the set $\Lambda = \{\lambda:  \lambda$  is any write operation that starts before $T_2 \text{ such that } tag(\lambda) > tag(\sigma^*)\}$. We define the number of writes concurrent with the valid read (repair) operation $\pi$ to be the cardinality of the set $\Lambda$.
				\end{definition}

				The above definition captures all the write operations that overlap with the read, until the time the reader has all data needed to attempt decoding a value. However, we ignore those write operations that might have started in the past, and never completed yet, if their tags are less than that of any write that completed before the read started. This allows us to ignore write operations due to failed writers, while counting concurrency, as long as the failed writes are followed by a successful write that completed before the read started. 
				
				The following lemma could be considered as the analogue of Lemma \ref{lem:property} for $\RADONC$. The first part of the  lemma shows that under $N1$ with $\alpha \geq \frac{3n+k}{4n}$, the repaired $List$ is never empty; there is always at least one (tag, coded-element) pair in the repaired $List$. Parts $2$ and $3$ are used to prove liveness and atomicity of client operations.

				\begin{lemma}\label{lem:property_randonc}
					Consider any well-formed execution $\beta$ of ${\RADONC}$ operating under  the network stability condition $N1$ with $\alpha \geq \frac{3n+k}{4n}$. Further assume that the number of  writes concurrent with any valid read or repair operation is at most $\delta$. For any operation $\pi$, consider the set $\Sigma = \{\sigma : \sigma~\text{is a read }$ $ \text{or a write in $\beta$ }$~$\text{that completes before}~\pi$ $\text{begins}\}$, and also let $\sigma^{*} = \arg\max\limits_{\sigma \in \Sigma}{\optag{\sigma}}$.  Then, the following statements hold: 
					\begin{itemize}
						\item If $\pi$ denotes a completed repair operation on a server $s \in \mathcal{S}$, then the repaired $List$ of server $s$ due to $\pi$ contains the pair $(tag(\sigma^*), c_s^*)$.
						\item If $\pi$ denotes a read operation associated with a non-faulty reader $r$, and further, if $\mathcal{S}_1$ denotes the set of  $\left\lceil\frac{n+k}{2}\right\rceil$ servers whose responses, say $\{L_{\pi}(s), s \in \mathcal{S}_1\}$, are used by $r$ to attempt decoding of a value in the {\GetData} phase, then there exists $\mathcal{S}_2 \subseteq \mathcal{S}_{1}$, $|\mathcal{S}_{2}| = k $, such that $\forall s \in \mathcal{S}_2,  (tag(\sigma^{*}), c_s^* ) \in L_{\pi}(s)$.
						\item If $\pi$ denotes a write operation associated with a non-faulty writer $w$, and further if $\mathcal{S}_1$ denotes the set of majority servers whose responses are used by $w$ to compute max-tag in the {\GetTag} phase, then there exists a server $s \in  \mathcal{S}_1$, whose response tag $t_s \geq tag(\sigma^*)$. 
					\end{itemize}
					Here, $c_s^*$ denotes the coded-element of server $s$ for value $v^*$, associated with $tag(\sigma^*)$.  
				\end{lemma}

				\begin{theorem}[Liveness]  \label{thm:liveness_radonc}
					Let $\beta$ denote a well-formed execution of $\RADONC$, operating under the $N1$  network  stability condition with $\alpha \geq \frac{3n+k}{4n}$ and $\delta$ be the maximum number of write operations concurrent with any valid read or repair operation. Then every operation initiated by a non-faulty client completes.
				\end{theorem}

				\begin{theorem}[Atomicity]  \label{thm:atomicity_radonc}
					Any execution of $\RADONC$, operating under condition $N1$ with $\alpha \geq \frac{3n+k}{4n}$, is atomic, if the maximum number of write operations concurrent with a valid read or repair operation is $\delta$.
				\end{theorem}

\vspace{-0.1in}		
		
\section{The ${\RADONS}$ Algorithm} \label{sec:radons}
	    In this section, we  present the  $\RADONS$~algorithm having the property that every execution is atomic. 
		Liveness is guaranteed under the slightly stronger network stability condition $N2$ with $\alpha > \frac{3}{4}$. In comparison wtih $\RADONL$, the algorithm has extra phases for both read and write operations, in order to guarantee safety of every execution. 
		
		\begin{algorithm*}[!ht]
			\begin{algorithmic}
				\begin{multicols}{2}{\footnotesize
						\Part{write(v)} { }\EndPart
						\Part{\underline{\GetTag}} {
							\State  $group\act{-}send(\QueryTag)$
							\State  Await responses from majority 
							\State  Select the max tag  $t^*$
						}\EndPart
						
						\Part{\underline{\PutData}} {
							\State $t_w = (t^{*}.z + 1, w)$.  
							\State $group\act{-}send( (\PutDataTag, (t_w, v)))$
							\State Wait for $\left\lceil \frac{3n + 1}{4}\right\rceil$ acks (say from $\mathcal{S}_{\alpha}$)
						}	\EndPart
						
						\Part{\underline{\ConfirmData}} {
							\State $group\act{-}send((\ConfirmDataTag, t_w))$
							\State Terminate after acks from majority from among servers in $\mathcal{S}_{\alpha}$
						}	\EndPart				
						
						\Statex
						\Part{read} { }\EndPart
						\Part{\underline{\GetData}} {
							\State  $group\act{-}send(\QueryTagData)$
							\State  Await responses from  majority
							\State  Select  $(t_r, v_r)$, with  max tag.
						}\EndPart	
						
						\Part{\underline{{\PutData} }}{
							\State $group\act{-}send((\PutDataTag, (t_r, v_r)))$
							\State Wait for $\left\lceil \frac{3n + 1}{4} \right\rceil$ acks (say from $\mathcal{S}_{\alpha}$)
						}\EndPart
						
						\Part{\underline{\ConfirmData}} {
							\State $group\act{-}send((\ConfirmDataTag, t_r))$
							\State Await acks from a majority of servers in $\mathcal{S}_{\alpha}$
							\State Return $v_r$
						}	\EndPart
						
						\Statex
						\Part{Server $s \in \mathcal{S}$}\EndPart
						\Part {\underline{$State~Variables$}}{ 
							\Statex $(t_{loc}, v_{loc}) \in \mathcal{T} \times {\mathcal V}$, initially   $(t_0, v_0)$
							\Statex $status \in \{active, repair\}$, initially $active$
							\Statex $Seen \subseteq {\mathcal T} \times \{{\mathcal W} \cup {\mathcal R}\}$, initially empty
							
						}\EndPart
						
						\Part {\underline{\GetTagResp, recv $\QueryTag$ from writer $w$}} {
							\If{ $status = active$ }
							\State Send  $t_{loc}$ to $w$
						 \EndIf
						}\EndPart
						
						\Part {\underline{\GetDataResp, recv $\QueryTagData$ from reader $r$}} {
							\If{ $status = active$ }
							\State Send  $(t_{loc}, v_{loc})$ to $r$
							\EndIf
						}\EndPart
						%								\Statex
						\Part{ \underline{\PutDataResp,~ recv $(\PutDataTag, (t, v))$ from $c$ }}{
							\If{$status = active$} 
							\If{ $t > t_{loc}$ } 
							\State  $(t_{loc}, v_{loc}) \leftarrow (t, v)$
							
							\EndIf 
							\State $Seen \leftarrow Seen \cup \{ (t, c) \}$
							\State  Send ack to $c$.
							\EndIf
							
						}\EndPart
						\Statex
						\Part{ \underline{\ConfirmDataResp, {\bf recv} $(\ConfirmDataTag, t)$ {\bf from }$c$}}
						 {%\small %(recv. $(\ConfirmDataTag, t)$  from $c$) }{
							\If{$status = active$} 
							\If{ $(t, c) \in Seen$ } 
							\State Remove $(t, c)$ from $Seen$ 
							\State Send  ack to client $c$.
							\EndIf 
							\EndIf
							
						}\EndPart
													\Statex
						\Part{ \underline{\InitRep} }{
							\State $status \leftarrow repair$
							\State $(t_{loc}, v_{loc}) \leftarrow (t_0, v_0)$
							\State $Seen \leftarrow \emptyset$ 
							\State $group\act{-}send(\RepairTagData)$
							\State Await responses from majority.
							\State Select $(t_{rep}, v_{rep})$, with max tag
							\State $(t_{loc}, v_{loc}) \leftarrow (t_{rep}, v_{rep})$
							\State $status \leftarrow active$ 
						}\EndPart
						\Statex
						\Part{ \underline{\InitRepResp,~recv $\RepairTagData$ from $s'$}}{
							\If{$status=active$} 
							\State  Send $(t_{loc}, v_{loc})$ to $s'$
							\EndIf
						}\EndPart
					}\end{multicols}
				\end{algorithmic}	
				\caption{The protocols for writer, reader, and any  server $s \in {\mathcal S}$ in  
					$\RADONS$.}\label{fig:radons-server}
			\end{algorithm*}				
			%\vspace{-0.15in}					
			The write operation has three phases (see Fig. \ref{fig:radons-server}). The first two phases are identical to those of $\RADONL$ during which the writer queries for the local tags, and then sends out the new (tag, value) pair, respectively. In the third phase, called {\ConfirmData}, the writer ensures the presence of at least a majority of servers, which the writer knows for sure that received its data during the second phase, {\PutData}. In order to facilitate the  {\ConfirmData}~phase, the servers maintain a $Seen$ variable. Any time the server receives a value from a writer, the server adds the corresponding (tag, writer ID) pair to the $Seen$ list. Next, during the {\ConfirmDataResp}~phase, the server responds to the writer only if this (tag, writer ID) pair is part of the $Seen$ variable. The idea is that if the server experiences a crash and a successful repair operation in between the {\PutData}~ and {\ConfirmData}~ phases, the server no longer has the (tag, writer ID) pair in its $Seen$ variable, and hence does not respond to the {\ConfirmData}~ phase. This is because, a crash removes all state variables, including $Seen$, and the repair algorithm (see Fig. \ref{fig:radons-server}) simply restores the $Seen$ variable to its default value, the empty set. Further, by ensuring that the writer expects acks from among a majority of servers in  {\ConfirmData}, from among the $\frac{3n+1}{4}$ servers whose acks were obtained during {\PutData}, we can guarantee that any execution is atomic.
			
			The read operation also has three phases, first two of which are identical to those of $\RADONL$, except for the use of the $Seen$ variable in the server during the {\PutData} phase. The third phase is the {\ConfirmData}~phase as in the write operation. The repair operation has one phase, and is nearly exactly identical to that of $\RADONL$. Note that the $Seen$ variable gets reset to its initial value during repair.
			
			\vspace{-0.1in}
			\subsection{Analysis of $\RADONS$} \label{sec:analysis_randons}	                             
			\vspace{-0.1in}	
			We overview the proofs of liveness and atomicity before formal claims. For liveness of writes, we assume $N2$ with $\alpha > \frac{3}{4}$, and argue the existence of a majority $\mathcal{S}_m$ of servers all of which remain active from the point of time at which the $group\act{-}send$ operation gets initiated in  the {\PutData} phase, till the point of time all the servers in $\mathcal{S}_m$ effectively consume requests for {\ConfirmData} from the writer.  In this case, write operation completes after receiving acks from servers in 
			$\mathcal{S}_m$ during the {\ConfirmData}~phase. The set $\mathcal{S}_m$ exists because, under $N2$ with $\alpha > \frac{3}{4}$,  a set $\mathcal{S}_{\alpha}$ of $\left\lceil \frac{3n + 1}{4} \right\rceil$ servers remain alive  from the start of the group-send, till the effective consumption of the acks by the writer in {\PutData} phase. Also, a second set $\mathcal{S}'_{\alpha}$ of $\left\lceil \frac{3n + 1}{4} \right\rceil$ servers remain active  from the start of the group-send in the {\ConfirmData} phase, till all servers in $\mathcal{S}_{\alpha}'$ complete the respective effective consumption from this group-send. We note that  $\mathcal{S}_{\alpha}' \cap \mathcal{S}_{\alpha}$ is at least a majority. We next use the observation that the $group\act{-}send$ operation in the {\ConfirmData} phase forms part of the effective consumption of the last of the acks in the {\PutData} phase. Using this, we argue that the servers in $\mathcal{S}_{\alpha}' \cap \mathcal{S}_{\alpha}$ remain active till they effectively consume message from $group\act{-}send$ operation of the {\ConfirmData} phase, and thus $\mathcal{S}_{\alpha}' \cap \mathcal{S}_{\alpha}$ is a candidate for $\mathcal{S}_m$. The liveness of read is similar to that of write, while liveness of repair is straightforward under $N2$ with $\alpha > \frac{3}{4}$.
			
			Towards proving atomicity of reads and writes, we first define tags for completed reads, writes and repair operations exactly in the same manner as we did in $\RADONL$.  Consider two completed write operations $\pi_1$ and $\pi_2$ such that $\pi_2$ starts after the completion of $\pi_1$, and we need to show that $tag(\pi_2) > tag(\pi_1)$.  As in $\RADONL$, we do this in two parts: Lemma \ref{lem:property} holds as it is for $\RADONS$ as well. Recall that Lemma \ref{lem:property} essentially shows that if a majority of active nodes is locked-on to any particular tag, say $t'$, at a specific point of time $T$ during the execution of the algorithm, then any repair operation which begins after the time $T$ always restores the tag to one which is at least as high as $t'$.  The challenge now is to show the existence of these favorable points of time instants $T$ as needed in the assumption of the lemma. While in $\RADONL$, we used the $N1$ to argue this, in $\RADONS$, we do not use $N2$; instead we rely on the third {\ConfirmData} phase of the first write operation $\pi_1$.  
			
			\begin{theorem}[Liveness]  \label{thm:liveness_radons}
				Let $\beta$ denote a well-formed execution of $\RADONS$  under condition $N2$  with $\alpha > \frac{3}{4}$ . Then every operation initiated by a non-faulty client completes.
			\end{theorem}

			\begin{theorem}[Atomcity]  \label{thm:atomicity_radons}
				Every execution of the $\RADONS$ algorithm is atomic.
			\end{theorem}

\section{Storage and Communication Costs of Algorithms} \label{sec:costs}

We give a justification of storage and communication cost numbers of the three algorithms, appearing in Table \ref{table:summ}. Recall that the size of value $v$ is assumed to be $1$ and also that we ignore the costs due to metadata. It is clear that both $\RADONL$ and $\RADONS$ have storage cost $n$, write cost $n$, and read cost $2n$ (due to write back). For $\RADONC$, each server stores at most $\delta+1 $ coded-elements, where each element has size $\frac{1}{k}$. Thus storage cost of $\RADONC$ is $(\delta + 1)\frac{n}{k}$. The write cost of $\RADONC$ is simply $\frac{n}{k}$, and the contribution comes from the writer sending one coded-element to each of the $n$ servers. For a read, getting  the entire $Lists$ during the $\GetData$ phase incurs a cost of $(\delta + 1)\frac{n}{k}$. The write-back phase incurs an additional cost of $\frac{n}{k}$. Thus, the total read cost in $\RADONC$ is $(\delta + 2)\frac{n}{k}$.
				
\vspace{-0.1in}
				
\section{Conclusions}\label{sec:conclusion}
In this paper, we provided an erasure-code-based algorithm for implementing atomic memory, having the ability to perform repair of crashed nodes in the background, without affecting client operations.  We assumed a static model with a fixed, finite set of nodes, and also a practical network condition $N1$ to facilitate repair. We showed how the usage of MDS codes significantly improve storage and communication costs over a replication based solution, when the number of writes concurrent with a read or repair is limited. Liveness and atomicity are guaranteed as long as $N1$ is satisfied; however violation of $N1$ can lead to non-atomic executions. We further showed how a slightly stringent network condition $N2$ can be used to construct a replication based algorithm that always guarantees atomicity. Ongoing efforts include exploring possibility of using repair-efficient erasure codes~\cite{dimakis2011survey} in $\RADONC$, and testbed evaluations on cloud based infrastructure. 

\section{Acknowledgments}

The work is supported in part by AFOSR  under grants FA9550-14-1-043, FA9550-14-1-0403, and in part by NSF under awards CCF-1217506, CCF-0939370.				
%%%%%%%%%%%%%%%%%%%%%%%%%%%%%%%%%%%%%%%%

{
\bibliographystyle{IEEEtran}
\bibliography{biblio}

\begin{thebibliography}{10}
\providecommand{\url}[1]{#1}
\csname url@samestyle\endcsname
\providecommand{\newblock}{\relax}
\providecommand{\bibinfo}[2]{#2}
\providecommand{\BIBentrySTDinterwordspacing}{\spaceskip=0pt\relax}
\providecommand{\BIBentryALTinterwordstretchfactor}{4}
\providecommand{\BIBentryALTinterwordspacing}{\spaceskip=\fontdimen2\font plus
\BIBentryALTinterwordstretchfactor\fontdimen3\font minus
  \fontdimen4\font\relax}
\providecommand{\BIBforeignlanguage}[2]{{%
\expandafter\ifx\csname l@#1\endcsname\relax
\typeout{** WARNING: IEEEtran.bst: No hyphenation pattern has been}%
\typeout{** loaded for the language `#1'. Using the pattern for}%
\typeout{** the default language instead.}%
\else
\language=\csname l@#1\endcsname
\fi
#2}}
\providecommand{\BIBdecl}{\relax}
\BIBdecl

\bibitem{ABD96}
H.~Attiya, A.~Bar-Noy, and D.~Dolev, ``Sharing memory robustly in message
  passing systems,'' \emph{Journal of the ACM}, vol. 42(1), pp. 124--142, 1996.

\bibitem{FL03}
R.~Fan and N.~Lynch, ``Efficient replication of large data objects,'' in
  \emph{Distributed algorithms}, ser. Lecture Notes in Computer Science, 2003,
  pp. 75--91.

\bibitem{kedar_bounds}
A.~{Spiegelman}, Y.~{Cassuto}, G.~{Chockler}, and I.~{Keidar}, ``{Space Bounds
  for Reliable Storage: Fundamental Limits of Coding},'' in \emph{Proceedings
  of the International Conference on Principles of Distributed Systems
  (OPODIS2015)}, 2015.

\bibitem{aguilera}
M.~K. Aguilera, R.~Janakiraman, and L.~Xu, ``Using erasure codes efficiently
  for storage in a distributed system,'' in \emph{Proceedings of International
  Conference on Dependable Systems and Networks (DSN)}, 2005, pp. 336--345.

\bibitem{DGL08}
P.~Dutta, R.~Guerraoui, and R.~R. Levy, ``Optimistic erasure-coded distributed
  storage,'' in \emph{Proceedings of the 22nd international symposium on
  Distributed Computing (DISC)}, Berlin, Heidelberg, 2008, pp. 182--196.

\bibitem{CadambeLMM14}
V.~R. Cadambe, N.~A. Lynch, M.~M{\'{e}}dard, and P.~M. Musial, ``A coded shared
  atomic memory algorithm for message passing architectures,'' in
  \emph{Proceedings of 13th {IEEE} International Symposium on Network Computing
  and Applications (NCA)}, 2014, pp. 253--260.

\bibitem{SODA2016}
K.~M. Konwar, N.~Prakash, E.~Kantor, N.~Lynch, M.~Medard, and A.~A.
  Schwarzmann, ``Storage-optimized data-atomic algorithms for handling erasures
  124 and errors in distributed storage systems,'' in \emph{30th IEEE
  International Parallel \& Distributed Processing Symposium (IPDPS)}, 2016.

\bibitem{amnesic}
R.~Guerraoui, R.~R. Levy, B.~Pochon, and J.~Pugh, ``The collective memory of
  amnesic processes,'' \emph{ACM Trans. Algorithms}, vol.~4, no.~1, pp. 1--31,
  2008.

\bibitem{Bigtable}
F.~Chang, J.~Dean, S.~Ghemawat, W.~C. Hsieh, D.~A. Wallach, M.~Burrows,
  T.~Chandra, A.~Fikes, and R.~E. Gruber, ``Bigtable: A distributed storage
  system for structured data,'' \emph{ACM Trans. Comput. Syst.}, vol.~26,
  no.~2, pp. 4:1--4:26, jun 2008.

\bibitem{DeCandia:2007}
G.~DeCandia, D.~Hastorun, M.~Jampani, G.~Kakulapati, A.~Lakshman, A.~Pilchin,
  S.~Sivasubramanian, P.~Vosshall, and W.~Vogels, ``Dynamo: Amazon's highly
  available key-value store,'' in \emph{Proceedings of Twenty-first ACM SIGOPS
  Symposium on Operating Systems Principles}, ser. SOSP '07.\hskip 1em plus
  0.5em minus 0.4em\relax New York, NY, USA: ACM, 2007, pp. 205--220.

\bibitem{LS02}
N.~Lynch and A.~A. Shvartsman, ``{RAMBO}: A reconfigurable atomic memory
  service for dynamic networks,'' in \emph{Proceedings of 16th International
  Symposium on Distributed Computing (DISC)}, 2002, pp. 173--190.

\bibitem{AKMS11}
M.~K. Aguilera, I.~Keidar, D.~Malkhi, and A.~Shraer, ``Dynamic atomic storage
  without consensus,'' \emph{Journal of the ACM}, pp. 7:1--7:32, 2011.

\bibitem{Alex_Idit}
\BIBentryALTinterwordspacing
A.~Spiegelman and I.~Keidar, ``On liveness of dynamic storage,'' \emph{CoRR},
  vol. abs/1507.07086, 2015. [Online]. Available:
  \url{http://arxiv.org/abs/1507.07086}
\BIBentrySTDinterwordspacing

\bibitem{BBKR09}
R.~Baldoni, S.~Bonomi, A.~M. Kermarrec, and M.~Raynal, ``Implementing a
  register in a dynamic distributed system,'' in \emph{Distributed Computing
  Systems, 2009. ICDCS '09. 29th IEEE International Conference on}, June 2009,
  pp. 639--647.

\bibitem{AttiyaCEKW15}
H.~Attiya, H.~C. Chung, F.~Ellen, S.~Kumar, and J.~L. Welch, ``Simulating a
  shared register in an asynchronous system that never stops changing -
  (extended abstract),'' in \emph{Distributed Computing - 29th International
  Symposium, {DISC} 2015, Tokyo, Japan, October 7-9, 2015, Proceedings}, 2015,
  pp. 75--91.

\bibitem{Aug_tutorial}
M.~K. Aguilera, I.~Keidar, D.~Malkhi, J.~P. Martin, and A.~Shraery,
  ``Reconfiguring replicated atomic storage: A tutorial,'' \emph{Bulletin of
  the EATCS}, vol. 102, pp. 84--081, 2010.

\bibitem{Spi_tutorial}
A.~Spiegelman, I.~Keidar, and D.~Malkhi, ``Dynamic reconfiguration: A
  tutorial,'' \emph{OPODIS 2015}, 2015.

\bibitem{Mus_tutorial}
P.~Musial, N.~Nicolaou, and A.~A. Shvartsman, ``Implementing distributed shared
  memory for dynamic networks,'' \emph{Communications of the ACM}, vol.~57,
  no.~6, pp. 88--98, 2014.

\bibitem{dimakis2011survey}
A.~G. Dimakis, K.~Ramchandran, Y.~Wu, and C.~Suh, ``A survey on network codes
  for distributed storage,'' \emph{Proceedings of the IEEE}, vol.~99, no.~3,
  pp. 476--489, 2011.

\bibitem{HuaSimXu_etal_azure}
C.~Huang, H.~Simitci, Y.~Xu, A.~Ogus, B.~Calder, P.~Gopalan, J.~Li, and
  S.~Yekhanin, ``{Erasure coding in windows azure storage},'' in \emph{{Proc.
  USENIX Annual Technical Conference (ATC)}}, 2012, pp. 15--26.

\bibitem{sathiamoorthy}
M.~Sathiamoorthy, M.~Asteris, D.~Papailiopoulos, A.~G. Dimakis, R.~Vadali,
  S.~Chen, and D.~Borthakur, ``{XORing elephants: novel erasure codes for big
  data},'' in \emph{Proceedings of the 39th international conference on Very
  Large Data Bases}, 2013, pp. 325--336.

\bibitem{rashmi_fast15}
K.~V. Rashmi, P.~Nakkiran, J.~Wang, N.~B. Shah, and K.~Ramchandran, ``Having
  your cake and eating it too: Jointly optimal erasure codes for i/o, storage,
  and network-bandwidth,'' in \emph{13th USENIX Conference on File and Storage
  Technologies (FAST)}, 2015, pp. 81--94.

\bibitem{rscodes}
I.~S. Reed and G.~Solomon, ``Polynomial codes over certain finite fields,''
  \emph{Journal of the society for industrial and applied mathematics}, vol.~8,
  no.~2, pp. 300--304, 1960.

\bibitem{cachin}
C.~Cachin and S.~Tessaro, ``Optimal resilience for erasure-coded byzantine
  distributed storage,'' in \emph{Proceedings of International Conference on
  Dependable Systems and Networks (DSN)}, 2006, pp. 115--124.

\bibitem{dobre}
D.~Dobre, G.~Karame, W.~Li, M.~Majuntke, N.~Suri, and M.~Vukoli{\'c},
  ``Powerstore: proofs of writing for efficient and robust storage,'' in
  \emph{Proceedings of the 2013 ACM SIGSAC conference on Computer \&
  communications security}, 2013, pp. 285--298.

\bibitem{hendricks}
J.~Hendricks, G.~R. Ganger, and M.~K. Reiter, ``Low-overhead byzantine
  fault-tolerant storage,'' \emph{ACM SIGOPS Operating Systems Review},
  vol.~41, no.~6, pp. 73--86, 2007.

\bibitem{regular_lamport}
L.~Lamport, ``On interprocess communication,'' \emph{Distributed computing},
  vol.~1, no.~2, pp. 86--101, 1986.

\bibitem{shao}
C.~Shao, J.~L. Welch, E.~Pierce, and H.~Lee, ``Multiwriter consistency
  conditions for shared memory registers,'' \emph{SIAM Journal on Computing},
  vol.~40, no.~1, pp. 28--62, 2011.

\bibitem{RADON:arxiv}
K.~M. Konwar, N.~Prakash, M.~Medard, and N.~Lynch, ``{RADON}: Repairable atomic
  data object in networks,'' \emph{CoRR}, vol. abs/1605.05717, 2016.

\bibitem{verapless_book}
W.~C. Huffman and V.~Pless, \emph{Fundamentals of error-correcting
  codes}.\hskip 1em plus 0.5em minus 0.4em\relax Cambridge university press,
  2003.

\bibitem{Lynch1996}
N.~A. Lynch, \emph{Distributed Algorithms}.\hskip 1em plus 0.5em minus
  0.4em\relax Morgan Kaufmann Publishers, 1996.

\end{thebibliography}
}

\appendix
\section*{Appendix}
\section{Proof of Theorem \ref{thm:imp_weak}} \label{app:imp_weak}

The theorem is restated for convenience.

\begin{theorem} ({\bf Theorem~\ref{thm:imp_weak}})
	{	It is impossible to implement an atomic memory service that guarantees liveness of reads and writes,  under the system model described in Section \ref{sec:models}, even if $1$) there is at most one server in the crashed/repair state at any point during the execution, and $2$) every repair operation completes, and takes the repaired server back to the active state.}
\end{theorem}
\begin{proof}
	We prove this result by contradiction, by assuming an algorithm $A_{alg}$ that guarantees liveness and atomicity, and also is such that every  repair operation completes, and takes the repaired server back to the active state.  Let the initial value stored in the system be $v_o \in V$, where $V$ is the domain of all values. Consider a non-faulty writer  $w$, and suppose $w$ initiates a write operation $\pi^{w}$ with the value $v_1$, such that $v_1 \neq v_0$. Let $\mathcal{S}_w \subseteq \mathcal{S}$ be the set of all servers that writer $w$ sends messages before $w$ expects any response from any of the servers  in $\mathcal{S}$. Without loss of generality\footnote{Clearly, the writer must send a message to at least one server, so we ignore the trivial case when $\mathcal{S}_w$ is empty.} let $\mathcal{S}_w = \{s_1, s_2, \ldots, s_k\}$, for some $k \leq n$, and let $m_i$ denote the message sent by $w$ to server $s_i, 1 \leq i \leq k$. Note that if $w$ sends two or messages to a particular server, say $s_1$, all these can be combined into $m_1$, since all these messages are sent without expecting any response. 
	
	Consider an execution which starts with all the servers in the active state, the operation $\pi^{w}$ begins, messages get sent out to servers in  $\mathcal{S}_w$. Delay the messages such that message $m_1$ arrives at server $s_1$ before any other server in $\mathcal{S}_w$ receives the respective message. Assume that $s_1$ is in the crashed state when $m_1$ arrives, so $s_1$ does not receive $m_1$. Further assume that all the other servers are in the active state at this point of execution. Let server $s_1$ undergo a successful repair operation, before any other server in $\mathcal{S}_w$ receives its respective message. Next, consider the case when server $s_2$ receives the message $m_2$, and delay the messages to all other servers, and assume that $s_2$ is in the crashed state when $m_2$ arrives. The sequence of crash and repair can be carried out in this manner one-by-one for every server in $\mathcal{S}_w$, where all these servers end up losing the writer message, though they get repaired. Now if the algorithm is such that the writer expects a response from any of the servers in $\mathcal{S}$, clearly it will not happen, since no server in $\mathcal{S}$ has received any message from $w$ while the server is in the active state. Thus liveness of write is compromised. 
	
	We next consider the case when the writer decides to terminate without expecting any response from any server in $\mathcal{S}$, and show that such a method of guaranteeing liveness results in violation of atomicity. Let us call this execution fragment (as discussed above) with such a write as  $\beta^{w}(v_1)$. After the write $\pi_w$ completes, a read $\pi_r$ associated with a non-faulty reader, begins. By liveness of read, and atomicity, the read must return $v_1$. Let the execution fragment associated with the read be denoted as $\beta^r$, so that the overall execution fragment under consideration is $\beta^{w}(v_1) \circ \beta^r$. Next, consider the execution fragment $\beta^{w}({v'}_1)$ obtained by replacing $v_1$ with $v_1'$ such that ${v_1}' \neq v_1$. Since a crash causes a server to lose its entire state, it is clear that to the reader $r$ there is no distinction between the state of the system 
	after $\beta^{w}(v_1)$, and the state of the system after $\beta^{w}(v_1')$. In this case, if we consider the execution $\beta^{w}(v_1') \circ \beta^r$, the read returns $v_1$ ($\neq v_1'$), since in the execution $\beta^{w}(v_1) \circ \beta^r$ also, $r$ returned $v_1$. However it violates atomicity of $\beta^{w}(v_1') \circ \beta^r$, which completes the proof.
\end{proof}

\section{Proof of Lemma \ref{lem:property}} \label{app:property}
The lemma is restated for convenience.
\begin{lemma}[Lemma \ref{lem:property}] 
	Let $\beta$ denote a well-formed execution of the $\RADONL$ algorithm. Suppose $T$ denotes a point of time in the execution $\beta$ such that there exists a majority of servers $\mathcal{S}_m$, $\mathcal{S}_m \subset \mathcal{S}$ all of which are in the active state at the time $T$. Also, let $t_s$ denote the value of the local tag at server $s$, at time $T$. Then, if $\pi$ denotes any completed repair or read operation that is initiated after time $T$, we have $tag(\pi) \geq \min_{s\in \mathcal{S}_m} t_s$. Also, if $\pi$ denotes any completed write  operation that is initiated after time $T$, then we have $tag(\pi) > \min_{s\in \mathcal{S}_m} t_s$.
\end{lemma}
\begin{proof}
	We use $\rho$ to denote $\min_{s\in \mathcal{S}_m} t_s$. Also, for any state variable $x(s)$ that is stored in server $s$, we write $x(s)|_{T}$ to denote the value $x$ at time $T$. Below, we separately consider the cases when $\pi$ denotes a successful repair, read and write operations, in this respective order.
	
	\noindent {\bf (a) $\pi$ is a successful repair operation:} We prove the statement by contradiction, by
	starting with the assumption that $tag(\pi) < \rho$. Let $T_{\pi}$ denote the point of time in the execution $\beta$ at which the operation $\pi$ completes. Let $\Pi'_R$ denote the set of all successful repair operations which start after the time $T$, but start before $T_{\pi}$, and is such that $\forall \pi' \in \Pi'_R$, we have $tag(\pi') < \rho$. Clearly, $\pi \in \Pi'_R$. Let $\pi^* \in \Pi'_R$ denote the repair operation, which completes first. Note that $\pi^*$ exists since the set $\Pi'_R$ is finite.
	Now, let $\hat{\mathcal{S}}$ denote the set of majority servers based on whose responses the operation $\pi^*$ completed. Clearly, $|\hat{\mathcal{S}} \cap \mathcal{S}_{m}| \geq 1$. For any server $s \in \hat{\mathcal{S}} \cap \mathcal{S}_{m}$, let $T_s$ denote the point of time in the execution at which the server $s$ responds to $\pi^*$ with its local (tag, value) pair. Clearly, the server $s$ must have remained in the active state during the entire interval $[T, T_s]$. This follows because $s$ is active at time $T$,  $\pi^*$ is the first completed repair operation that started after $T$, and due to the fact that a server responds to a repair request only if it is in the active state. In this case, we know that\footnote{Any read or write operation cannot decrease the local tag that is stored in an active server.} $t_{loc}(s)|_{T_s} \geq t_{loc}(s)|_T \geq \rho$ for any $s$ in $\hat{\mathcal{S}} \cap \mathcal{S}_{m}$. Therefore, we have $tag(\pi^*) = \max_{s \in \hat{\mathcal{S}}} t_{loc}(s)|_{T_s} \geq \max_{s \in \hat{\mathcal{S}} \cap \mathcal{S}_m} t_{loc}(s)|_{T_s} \geq \rho$, which contradicts the existence of $\pi^* \in \Pi'_R$. From, this we conclude that the set $\Pi'_R$ must be empty to avoid contradictions, and hence $tag(\pi) \geq \rho$.
	
	\noindent {\bf (b) $\pi$ is a successful read operation:} We prove this by contradiction by starting with the assumption that $tag(\pi) < \rho$.  Let $\hat{\mathcal{S}}$ denote the set of majority servers based on whose responses during the $get$-$data$ phase (see Fig. \ref{fig:radonl-server}), the read operation completed. As in Part $a)$, we know that $|\hat{\mathcal{S}} \cap \mathcal{S}_{m}| \geq 1$. In this case, let $T_s$ denote the point of time during the execution at which the server $s \in \hat{\mathcal{S}} \cap \mathcal{S}_{m}$ responded to the reader. Next, note that in the $get$-$data$ phase, the reader picks the response with the highest tag. Thus, since we assume that $tag(\pi) < \rho$, it must be true that $t_{loc}(s)|_{T_s} < \rho, s \in \hat{\mathcal{S}} \cap \mathcal{S}_{m}$. Since the server $s \in \hat{\mathcal{S}} \cap \mathcal{S}_{m}$ is active at time $T$ such that $t_{loc}(s)|_T \geq \rho$, this would imply that server $s$ experienced a crash event after time $T$, and  came back to the \emph{active} state before the time $T_s$ via a successful repair operation $\phi$ such that $tag(\phi) < \rho$. But then, this contradicts Part a) of the theorem which we proved above, and hence we conclude that $tag(\pi) \geq \rho$.
	
	\noindent {\bf (c) $\pi$ is a write operation:} Once again we prove via contradiction, by starting with the assumption that $tag(\pi) \leq \rho$. Let $\hat{\mathcal{S}}$ denote the set of majority servers based on whose responses during the $get$-$tag$ phase, the writer determined $tag(\pi)$. We know from the algorithm that  $tag(\pi)$ is strictly larger than all the tags among the responses from $\hat{\mathcal{S}}$.  Since $|\hat{\mathcal{S}} \cap \mathcal{S}_{m}| \geq 1$, we argue like in Part $b)$, and arrive at a contradiction to Part $a)$. 
\end{proof}
\section{Proof of Theorem \ref{thm:atomicity}} \label{app:atomicity}

The theorem is restated here for convenience. 

\begin{theorem}[Theorem \ref{thm:atomicity}]  \label{thm:atomicity_app}
	Every execution of the $\RADONL$ algorithm operating under the $N1$ network stability condition with $\alpha > \frac{3}{4}$, is atomic.
\end{theorem}

\subsection{Some Preliminaries}

\paragraph*{Partial Order on read and write operations}
Consider any well-formed execution $\beta$ of $\RADONL$, all of whose invoked read or write operations complete. Let $\Pi_{RW}$ denote the set of all completed read and write operations in $\beta$. We first define a partial order ($\prec$) on $\Pi_{RW}$. Towards this, recall that for any completed write operation $\pi$, we defined $tag(\pi)$ as the tag created by the writer during the {\PutW} phase. Also, recall that for any completed read operation $\pi$, we define $tag(\pi)$ as the tag corresponding to the value returned by the read. The partial order ($\prec$) in $\Pi_{RW}$ is defined as follows: For any $\pi, \phi \in \Pi_{RW}$, we say $\pi \prec \phi$  if one of the following holds: $(i)$  $tag(\pi)  < tag(\phi)$, or $(ii)$ $tag(\pi) = tag(\phi)$, and  $\pi$ and $\phi$ are write and read operations, respectively. The proof of atomicity is based on the following lemma, which is simply a restatement of the sufficiency condition for atomicity presented in~\cite{Lynch1996}.

\begin{lemma} \label{lem:atom}
	Consider any well-formed execution $\beta$ of the algorithm, such that all the invoked read and the 
	write operations  complete. Now, suppose that all  the invoked read and write operations in $\beta$ can be partially ordered by an ordering $\prec$, so that the following properties are satisfied:
	\begin{itemize}
		\item [\em P1.] The partial order ($\prec)$ is consistent with the 
		external order of invocation and responses, i.e., there are no 
		operations $\pi_1$ and $\pi_2$, 
		such that $\pi_1$ completes before $\pi_2$ starts, 
		yet $\pi_2 \prec \pi_1$.
		\item[\em P2.] All operations are totally 
		ordered with respect to the write operations, 
		i.e., if $\pi_1$ is a write operation and $\pi_2$ is any other operation then 
		either $\pi_1 \prec \pi_2$ or $\pi_2 \prec \pi_1$.
		\item[\em P3.] Every read operation returns
		the value of the last write preceding it (with respect to $\prec$), and if no preceding 
		write is ordered before it, then the read returns the initial value
		of the object.
	\end{itemize}
	Then, the execution $\beta$ is atomic.	
\end{lemma}	

\subsection{Proof of Atomicity under $N1$ with $\alpha > 3/4$}
We need to prove the properties $P1$, $P2$ and $P3$ of Lemma~\ref{lem:atom}. We do this under $N1$ with $\alpha > 3/4$, using Lemma \ref{lem:property}. Let $\phi$ and $\pi$ denote two operations in $\Pi_{RW}$ such that $\phi$ completed before $\pi$ started.  Also, let  $c_{\phi}$ and $c_{\pi}$ denote the clients that initiated the operations $\phi$ and $\pi$, respectively. 

\paragraph*{Property $P1$} We want to show that $\pi \not\prec \phi$. We show this in detail only for the case when $\phi$ and $\pi$ are  both write operations. The proofs of other three cases\footnote{These correspond to the case when $\phi$ and $\pi$ are  both read operations, and the cases where one of them is a write and the other is a read.} are similar, and hence omitted. By virtue of the definition of the partial order $\prec$, it is enough to prove that $tag(\pi) > tag(\phi)$. Consider the $put$-$data$ phase of $\phi$, where the writer sends the pair $(t_w, v)$ to all 
servers via the $group$-$send$ operation. Under the condition $N1$ with $\alpha > 3/4$, we know that there exists a set $\mathcal{S}_{\alpha} \subseteq \mathcal{S}$ of $\lceil n\alpha \rceil \geq  \left\lceil \frac{3n+1}{4} \right\rceil$ servers all of 
which remain in the active state during the interval $[T_1, T_2]$ where $T_1$ denotes 
the point of time of invocation of the group-send operation, and $T_2$ denotes the earliest point of time during the 
execution where all of the servers in $\mathcal{S}_{\alpha}$ complete effective consumption (including sending ack to the writer $c_{\phi}$) of the message $(t_w, v)$. Also, let $\mathcal{S}' \subseteq \mathcal{S}$ denote the set of $\left\lceil\frac{3n+1}{4}\right\rceil$ servers whose 
acks are used by the writer to decide the completion of the write operation. Clearly, $|\mathcal{S}' \cap \mathcal{S}_{\alpha}| > \frac{n}{2}$. Let $T$ denote 
the earliest point of time during the execution when all servers in $\mathcal{S}' \cap \mathcal{S}_{\alpha}$ complete their respective effective consumption of the message $(t_w, v)$. In this case note that $a)$ $T$ occurs before the point of completion of the write operation, $b)$ all servers in $\mathcal{S}' \cap \mathcal{S}_{\alpha}$ are in the active state at $T$, and  $c)$ $t_{loc}(s)|_T \geq tag(\phi), \forall s \in  \mathcal{S}' \cap \mathcal{S}_{\alpha}$.  We now apply Lemma  \ref{lem:property} to conclude that $tag(\pi) > tag(\phi)$.

\paragraph*{Property $P2$} This follows from the construction of tags, and the definition of the partial order ($\prec$).

\paragraph*{Property $P3$} This follows from the definition of partial order ($\prec$), and by noting that value returned by a read operation $\pi$ is simply the value associated with $tag(\pi)$.

\section{Proof of Lemma \ref{lem:property_randonc} } \label{app:property_randonc}

The lemma is restated below for easy reference. 

\begin{lemma}[Lemma $2$]\label{lem:property_randonc_rep}
	Consider any well-formed execution $\beta$ of ${\RADONC}$ operating under  the network stability condition $N1$ with $\alpha \geq \frac{3n+k}{4n}$. Further assume that the number of  writes concurrent with any valid read or repair operation is at most $\delta$. For any operation $\pi$, consider the set $\Sigma = \{\sigma : \sigma~\text{is a read }$ $ \text{or a write in $\beta$ }$~$\text{that completes before}~\pi$ $\text{begins}\}$, and also let $\sigma^{*} = \arg\max\limits_{\sigma \in \Sigma}{\optag{\sigma}}$.  Then, the following statements hold: 
	\begin{itemize}
		\item If $\pi$ denotes a completed repair operation on a server $s \in \mathcal{S}$, then the repaired $List$ of server $s$ due to $\pi$ contains the pair $(tag(\sigma^*), c_s^*)$.
		\item If $\pi$ denotes a read operation associated with a non-faulty reader $r$, and further, if $\mathcal{S}_1$ denotes the set of  $\left\lceil\frac{n+k}{2}\right\rceil$ servers whose responses, say $\{L_{\pi}(s), s \in \mathcal{S}_1\}$, are used by $r$ to attempt decoding of a value in the {\GetData} phase, then there exists $\mathcal{S}_2 \subseteq \mathcal{S}_{1}$, $|\mathcal{S}_{2}| = k $, such that $\forall s \in \mathcal{S}_2,  (tag(\sigma^{*}), c_s^* ) \in L_{\pi}(s)$.
		\item If $\pi$ denotes a write operation associated with a non-faulty writer $w$, and further if $\mathcal{S}_1$ denotes the set of majority servers whose responses are used by $w$ to compute max-tag in the {\GetTag} phase, then there exists a server $s \in  \mathcal{S}_1$, whose response tag $t_s \geq tag(\sigma^*)$. 
	\end{itemize}
	Here, $c_s^*$ denotes the coded-element of server $s$ corresponding to the value $v^*$, associated with $tag(\sigma^*)$.
\end{lemma}

We prove the lemma separately for the cases of repair and read (Parts $1$ and $2$). The proof for the third part for the case of write operations to similar to that of Part $2$, and hence omitted.

\subsubsection{Proof of Part $1$ of Lemma \ref{lem:property_randonc_rep}}

Consider the set $\Sigma$ and the operation $\sigma^*$ as defined in the statement of Lemma \ref{lem:property_randonc_rep}. Without loss of generality, let us assume that $\sigma^*$ is a write operation. Since we assume condition $N1$ with $\alpha \geq \frac{3n+k}{4n}$, there exists a set $\mathcal{S}_{\alpha}$ of $\left \lceil \frac{3n+k}{n} \right \rceil$ servers that respects $N1$ for the group-send operation (say gp*) in the {\PutData} phase of $\sigma^*$. If $\mathcal{S}_1$ denotes the set of $\left \lceil \frac{3n+k}{n} \right \rceil$ servers, whose responses are used by the writer to decide termination, we then know that $1)$ $|\mathcal{S}_{\alpha} \cap \mathcal{S}_1| \geq \lceil \frac{n+k}{2} \rceil$, and $2)$ if $T_{prop}$ denotes the earliest point of time during the execution when all the servers in $\mathcal{S}_{prop}  = \mathcal{S}_{\alpha} \cap \mathcal{S}_1$ complete effective consumption of their respective messages from the group-send operation gp*, then every server in $\mathcal{S}_{prop}$ remains active at $T_{prop}$, and has not experienced a crash after its effective consumption, until $T_{prop}$. Our goal is to show that the repair operation $\pi$ always receives at least $k$ responses from among the servers in $\mathcal{S}_{prop}$, and must be able to decode (and then re-encode) the value corresponding to $tag(\sigma^*)$. Below we consider the effects of concurrent writes having higher tags, and repairs before $\pi$ starts, both of which can potentially remove coded elements corresponding to $tag(\sigma^*)$, from lists of various servers. We show under the assumptions of the lemma, that neither of these cause a problem.

Let us first consider the effect of concurrent writes. Towards this, consider the set $\Lambda$ of writes concurrent with  the valid repair operation $\pi$ (see Definition \ref{defn:concurrent}). Recall that $\Lambda = \{\lambda: \lambda \text{ is a write operation that starts before completion of } $ $ \pi, \text{such that } tag(\lambda) > tag(\sigma^*)\}$.  By assumption on the lemma, we know that $|\Lambda| \leq \delta$.  In this case, it is clear that if a server $s \in \mathcal{S}_{prop}$ does not crash in the interval $[T_{prop} \ \ T]$, the $List(s)|_T$ contains the pair corresponding to $tag(\sigma^*)$, for any $T$ such that $T_{prop} \leq T \leq T_{end}(\pi)$. Here $T_{end}(\pi)$ denotes the point of completion of $\pi$. 

Let us next consider the effect of repairs, let $\widetilde{\Pi} = \{\widetilde\pi:$ a repair which start after $T_{prop}$, but also start before the completion of $\pi\}$. Clearly, $\pi \in \widetilde{\Pi}$. Let $\widetilde\pi^* \in \widetilde{\Pi}$ denote the repair operation that completes first. Clearly, it must be true that $T_{prop} < T_{end}(\widetilde{\pi}^*) \leq T_{end}(\pi)$. We prove Part $1$ of the Lemma \ref{lem:property_randonc_rep} for $\widetilde\pi^*$ first. Using this result, we prove the lemma for the repair operation in $\widetilde{\Pi}$ which completes second. We continue in an inductive manner (on the finite set $\widetilde{\Pi}$), until we hit $\pi$. Towards proving the lemma for $\widetilde\pi^*$, consider the group-send operation, where $\widetilde\pi^*$ requests for local $Lists$ from all servers. Let $\mathcal{S}_{\theta} \subset \mathcal{S}_{prop}$ denote the servers among $\mathcal{S}_{prop}$ which are not in the active state when the repair request arrives. Also, let $\mathcal{S}_a \subset \mathcal{S}$ denote the set of all servers which are in the active state when the repair request arrives. Clearly, $|\mathcal{S}_a| \leq n - |\mathcal{S}_{\theta}|$. Next, let $\mathcal{S}_{ack} \subset \mathcal{S}_a$ denote the set of $\lceil \frac{n+k}{2} \rceil$ servers based on whose responses the repair operation $\widetilde\pi^*$ completes. Now, since $\mathcal{S}_{prop}\backslash \mathcal{S}_{\theta} \subset  \mathcal{S}_a$, we have  
\begin{eqnarray}
(\mathcal{S}_{prop}\backslash \mathcal{S}_{\theta}) \cup \mathcal{S}_{ack} & \subset & \mathcal{S}_a \label{eq:0}\\
\implies |\mathcal{S}_{prop}\backslash \mathcal{S}_{\theta}| + |\mathcal{S}_{ack}| - |(\mathcal{S}_{prop}\backslash \mathcal{S}_{\theta}) \cap \mathcal{S}_{ack}| & \leq & \mathcal{S}_a \\
\implies |\mathcal{S}_{prop}| - |\mathcal{S}_{\theta}| + \lceil \frac{n+k}{2} \rceil - |(\mathcal{S}_{prop}\backslash \mathcal{S}_{\theta}) \cap \mathcal{S}_{ack}| & \leq & n - |\mathcal{S}_{\theta}| \\
\implies |(\mathcal{S}_{prop}\backslash \mathcal{S}_{\theta}) \cap \mathcal{S}_{ack}| \geq k, \label{eq:1}
\end{eqnarray}
where the last inequality follows from our earlier observation that $|\mathcal{S}_{prop}| \geq \lceil \frac{n+k}{2} \rceil$.  Next, note that any  server $s$ in $(\mathcal{S}_{prop}\backslash \mathcal{S}_{\theta}) \cap \mathcal{S}_{ack}$ remains active from $T_{prop}$ until the point when $s$ responds to the repair request from 
$\widetilde{\pi}^*$. This follows because of the facts that $1)$ $s$ is active at $T_{prop}$, $2)$ a server responds to a repair request only if it is in the active state, and $3)$ 
since $\widetilde{\pi}^*$ is the first repair operation that completes after $T_{prop}$. Also, recall our earlier observations that $1)$ if a server $s \in \mathcal{S}_{prop}$ does not crash in the interval $[T_{prop} \ \ T]$, then $List(s)|_T$ contains the pair corresponding to $tag(\sigma^*)$, for any $T$ such that $T_{prop} \leq T \leq T_{end}(\pi)$, and $2)$ $T_{prop} < T_{start}(\widetilde{\pi}^*) < T_{end}(\widetilde{\pi}^*) \leq T_{end}(\pi)$. In this case, we know that the responses of all the servers in $(\mathcal{S}_{prop}\backslash \mathcal{S}_{\theta}) \cap \mathcal{S}_{ack}$ to $\widetilde{\pi}^*$, contain the pair corresponding to $tag(\sigma^*)$. From \eqref{eq:1}, it follows that the repaired list for $\widetilde{\pi}^*$, before pruning to $(\delta + 1)$ entries, contains the pair corresponding to $tag(\sigma^*)$. Finally the fact that $tag(\sigma^*)$ is among the  highest $\delta+1$ tags, and hence part of the pruned list, follows from our earlier observations\footnote{Note that $\Lambda$ need not be the set of writes concurrent with $\widetilde{\pi}^*$. The above argument where we say that the pruned list, after the repair $\widetilde{\pi}^*$, is of size at most $\delta +1$ can be argued entirely based on $\Lambda$ itself.}  that $1)$ $|\Lambda| \leq \delta$, and $2)$ $T_{end}(\widetilde{\pi}^*) \leq T_{end}(\pi)$. This completes our proof of Part $1$ of Lemma \ref{lem:property_randonc_rep} for the repair operation $\widetilde{\pi}^*$.

We next prove the lemma for the repair operation $\pi_2 \in \widetilde{\Pi}$, which completes second. The proof is mostly identical, and we will only highlight the place where we use the result on $\widetilde{\pi}^*$. Clearly, since we carry out the induction only until we hit $\pi$, it must be true that $T_{prop} < T_{end}({\pi}_2) \leq T_{end}(\pi)$. Consider the group-send operation, where $\pi_2$ requests for local $Lists$ from all servers. Let $\mathcal{S}_{\theta}^{(2)} \subset \mathcal{S}_{prop}$ denote the servers among $\mathcal{S}_{prop}$ which are not in the active state when the repair request arrives. Also, let $\mathcal{S}_a^{(2)} \subset \mathcal{S}$ denote the set of all servers which are in the active state when the repair request arrives. As before, $|\mathcal{S}_a^{(2)}| \leq n - |\mathcal{S}_{\theta}^{(2)}|$. Next, let $\mathcal{S}_{ack}^{(2)} \subset \mathcal{S}_a^{(2)}$ denote the set of $\lceil \frac{n+k}{2} \rceil$ servers based on whose responses the repair operation $\pi_2$ completes. Along the lines of \eqref{eq:0}-\eqref{eq:1}, one can show that $|(\mathcal{S}_{prop}\backslash \mathcal{S}_{\theta}^{(2)}) \cap \mathcal{S}_{ack}^{(2)}| \geq k$. Next, if we consider the set $(\mathcal{S}_{prop}\backslash \mathcal{S}_{\theta}^{(2)}) \cap \mathcal{S}_{ack}^{(2)}$, at most one of the servers in this set would have undergone a crash after the time $T_{prop}$, and got repaired before the time the server responded to $\pi_2$. Note that more than one repair operation on $(\mathcal{S}_{prop}\backslash \mathcal{S}_{\theta}^{(2)}) \cap \mathcal{S}_{ack}^{(2)}$ cannot happen, since this will contradict the assumption that $\pi_2$ is the second repair operation to complete after $T_{prop}$. Further, if one repair operation among a server in $(\mathcal{S}_{prop}\backslash \mathcal{S}_{\theta}^{(2)}) \cap \mathcal{S}_{ack}^{(2)}$ has indeed occurred, this must be the operation $\widetilde{\pi}^*$ which we considered above. Further, we know that the repaired $List$ due to $\widetilde{\pi}^*$ contains the pair corresponding to $tag(\sigma^*)$. In other words, irrespective of whether one repair operation occurred among the servers in $(\mathcal{S}_{prop}\backslash \mathcal{S}_{\theta}^{(2)}) \cap \mathcal{S}_{ack}^{(2)}$, or not,  the responses of all the servers in $(\mathcal{S}_{prop}\backslash \mathcal{S}_{\theta}^{(2)}) \cap \mathcal{S}_{ack}^{(2)}$ contain the pair corresponding to $tag(\sigma^*)$. The rest of the proof is similar to that of $\widetilde{\pi}^*$, where we argue that the pruned list after repair contains the pair corresponding to $tag(\sigma^*)$. The rest of the induction is similar, and this completes the proof of Part $1$ of Lemma \ref{lem:property_randonc_rep}.

\subsubsection{Proof of Part $2$ of Lemma \ref{lem:property_randonc_rep}}

The proof follows mostly along the lines of proof of Part $1$ of the lemma. We will only highlight the main steps here. Consider the set $\Sigma$ and the operation $\sigma^*$ as defined in the statement of the lemma. Without loss of generality, let us assume that $\sigma^*$ is a write operation. Also, define the time $T_{prop}$ and the set $\mathcal{S}_{prop}$ exactly in the same way as what we defined in the proof of Part $1$ of the lemma. Let $T_1$ denote the earliest point of time during the execution when the reader receives  responses from all the servers in $\mathcal{S}_1$, where $\mathcal{S}_1$ is as defined in the statement of this lemma. Consider the set of writes concurrent with the valid read operation $\pi$. Recall from Definition \ref{defn:concurrent} that $\Lambda = \{\lambda: \lambda \text{ is a write operation which starts before time } T_1 \text{ such that } tag(\lambda) > tag(\sigma^*)\}$. From the assumption on concurrency in the lemma statement, we know that $|\Lambda| \leq \delta$. In this case, it is clear (like in the proof of Part $1$ above) that if a server $s \in \mathcal{S}_{prop}$ does not crash in the interval $[T_{prop} \ \ T_1]$, the $List(s)|_T$ contains the pair corresponding to $tag(\sigma^*)$, for any $T$ such that $T_{prop} \leq T \leq T_1$.  Now, if server $s \in \mathcal{S}_{prop}$ undergoes a crash and repair operation (say $\rho$) during the interval $[T_{prop}\  \ T]$ (so that it is active again at $T$), we can argue exactly like in the proof of Part $1$ above, and show that the repaired $List$ due to $\rho$ contains the pair corresponding to $tag(\sigma^*)$. This can be done by considering the set $\widetilde{\Pi} = \{\widetilde\pi:$ a repair which start after $T_{prop}$, but also start before $T_1 \}$, and applying induction on $\widetilde{\Pi}$ based on the order of completion times of the repair operations. The completes the proof of our claim about $List(s)|_T$.

The rest of the proof  follows simply by noting $|\mathcal{S}_1 \cap \mathcal{S}_{prop}| \geq k$, and thus the value corresponding to $tag(\sigma^*)$ is surely a candidate for decoding, since we know that an $[n, k]$ linear MDS code can be uniquely decoded given any $k$ out of the $n$ coded-elements.

\section{Liveness: Proof of Theorem \ref{thm:liveness_radonc}}\label{app:liveness_radonc}

The theorem is restated below for easy reference:

\begin{theorem}(Theorem \ref{thm:liveness_radonc})  \label{thm:liveness_radonc-app}
	Let $\beta$ denote a well-formed execution of $\RADONC$, operating under the $N1$  network  stability condition with $\alpha \geq \frac{3n+k}{4n}$ and $\delta$ be the maximum number of write operations concurrent with any valid read or repair operation. Then every operation initiated by a non-faulty client completes.
\end{theorem}
\begin{proof}
	Liveness of writes depends only on sufficient number of responses in the two phases. The maximum number of responses expected in either of the two phases is $\frac{3n+k}{4n}$, which we know is guaranteed under $N1$ with $\alpha \geq \frac{3n+k}{4n}$. Liveness of reads follows by combining Lemma \ref{lem:property_randonc} (for decodability of a value), and the liveness of write operations (for the write-back phase).
\end{proof}

\section{Atomicity: Proof of Theorem \ref{thm:atomicity_radonc}}\label{app:atomicity_radonc}

The theorem is restated first:

\begin{theorem}(Theorem \ref{thm:atomicity_radonc})  \label{thm:atomicity_radonc_app}
	Any execution of $\RADONC$, operating under condition $N1$ with $\alpha \geq \frac{3n+k}{4n}$, is atomic, if the maximum number of write operations concurrent with a valid read or repair operation is $\delta$.
\end{theorem}
\begin{proof}
	The proof is based on Lemmas~\ref{lem:atom} and  \ref{lem:property_randonc}. In order to apply Lemma \ref{lem:atom}, consider any well-formed execution $\beta$ of $\RADONC$, all of whose invoked read and write operations, denoted by the set $\Pi_{RW}$, complete.  We define a partial order ($\prec$) on $\Pi_{RW}$ like in the proof of Theorem \ref{thm:atomicity} for case of $\RADONL$. To prove Property $P1$ of Lemma \ref{lem:atom}, consider two successful operations $\phi$ and $\pi$ such  that $\phi$ completes before $\pi$ begins. Firstly, consider the case $\pi$ is a write, and $\phi$ is either a read or write. We need to show that $tag(\pi) > tag(\phi)$, which we note follows directly from Part $3$ of Lemma \ref{lem:property_randonc}. Next, consider the case when consider the case $\pi$ is a read, and $\phi$ is either a read or write. We need to show that $tag(\pi) \geq tag(\phi)$, which we note follows directly from Part $2$ of Lemma \ref{lem:property_randonc}. This completes the proof of Property $P1$. Proofs of Properties $P2$ and $P3$ are similar to those of the corresponding properties in Theorem \ref{thm:atomicity}, where we proved atomicity of $\RADONL$.
\end{proof}

\section{Proof of Theorem \ref{thm:liveness_radons}} \label{app:liveness_radons}

The theorem is restated for convenience. 

\begin{theorem}(Theorem \ref{thm:liveness_radons})  \label{thm:liveness_radons_app}
	Let $\beta$ denote a well-formed execution of $\RADONS$ operating under condition $N2$  with $\alpha > \frac{3}{4}$ . Then every operation initiated by a non-faulty client completes.
\end{theorem}

We will prove that a write operation associated with a non-faulty client always completes, the proof for a read is similar and hence is omitted. The main step is to show the completion of the {\ConfirmData} phase. Consider the {\PutData} phase, and note that under $N2$ with $\alpha > \frac{3}{4}$, we are guaranteed that there exists of set of $\mathcal{S}_{\alpha} \subset \mathcal{S}$ severs, such that $1)$~ $|\mathcal{S}_{\alpha}| \geq \left\lceil \frac{3n + 1}{4} \right\rceil$, and $2)$~every server in $ \mathcal{S}_{\alpha}$ remains active from the point of time $T_1$ of initiation of the group-send operation of {\PutData} phase till the point of time $T_1'$, when the writer effectively consumes all responses (acks) from the servers in  $\mathcal{S}_{\alpha}$. Next, let $\mathcal{S}_1 \subset \mathcal{S}$ denote the set of $\left\lceil \frac{3n + 1}{4} \right\rceil$ whose acks are received by the writer before moving on to the {\ConfirmData} phase. First of all note that the existence of the set $\mathcal{S}_1$ is clearly guaranteed under $N2$  with $\alpha > \frac{3}{4}$ (since the set $\mathcal{S}_{\alpha}$ is a candidate for $\mathcal{S}_1$).
Secondly, we note that the group-send operation in the {\ConfirmData} phase forms part of the effective consumption of the last ack that is received from the servers in $\mathcal{S}_1$. This follows from the definition of effective-consumption, and by noting the execution of the group-send operation in the {\ConfirmData} phase does not depend on any more input after all the acks in the {\PutData} phase are received. Let $T_2$ denote the point of time at which the group-send operation in the {\ConfirmData} phase gets initiated. Note that $T_1' \geq T_2$, in fact if $\mathcal{S}_{1} \neq \mathcal{S}_{\alpha}$, we have\footnote{In this case, some of the acks from the servers in $\mathcal{S}_{\alpha}$ get effectively consumed only after the required number $\left\lceil \frac{3n + 1}{4} \right\rceil$ have already been consumed, the last of which includes execution of the group-send operation of the {\ConfirmData}~phase. We note that the effective consumption of these additional acks from servers in $\mathcal{S}_{\alpha}$ is the operation where server simply ignores these, which is not explicitly mentioned in the algorithm. We also note that the notion of atomicity of any sequences of effective consumptions that are local to a server, is implicitly used when we argue that $T_1' > T_2$. By this we mean that if a server receives a message $m_1$ before $m_2$, the effective consumption of message $m_1$ is assumed to be entirely completed before the effective consumption of the message $m_2$ starts.} $T_1' > T_2$. Next we apply the network condition to the group-send operation in the {\ConfirmData}~phase. From the $N1$ part of $N2$, we know that there exists a $\mathcal{S}'_{\alpha}$ of $\left\lceil \frac{3n + 1}{4} \right\rceil$ servers, all of which receive and effectively consume the message from the group-send operation, and remain active from $T_2$ till the point of time $T_2'$ when the last of the servers in $\mathcal{S}'_{\alpha}$ completes effective consumption. Now if we let $\mathcal{S}_{\gamma} = \mathcal{S}_{\alpha} \cap \mathcal{S}'_{\alpha}$, observe that $1)$ $|\mathcal{S}_{\gamma}| > \frac{n}{2}$, and $2)$ all the servers in $\mathcal{S}_{\gamma}$ remain active from $T_1$ till $T_2'$. The second part follows from our earlier observation that $T_1' \geq T_2$. In this case, we infer that all the servers in $\mathcal{S}_{\gamma}$ does indeed acknowledge back to writer as part of their effective consumption of the {\ConfirmData} message, and since $\mathcal{S}_{\gamma} \subset \mathcal{S}_{\alpha}$ is at least a majority, we conclude that the write operation associated with the non faulty writer eventually completes.

\section{Proof of Theorem \ref{thm:atomicity_radons}} \label{app:atomicity_radons}

\begin{theorem}(Theorem \ref{thm:atomicity_radons})  \label{thm:atomicity_radons_app}
	Every execution of the $\RADONS$ algorithm is atomic.
\end{theorem}	

\subsection{Some Preliminaries}
The proof is based on Lemma~\ref{lem:atom}, and the equivalent of Lemma \ref{lem:property} for $\RADONS$, which we state below for the sake of completion:

\begin{lemma}  \label{lem:property_radons}
	Let $\beta$ denote a well-formed execution of $\RADONS$. Suppose $T$ denotes a point of time in $\beta$ such that there exists a majority of servers $\mathcal{S}_m$, $\mathcal{S}_m \subset \mathcal{S}$ all of which are in the active state at  time $T$. Also, let $t_s$ denote the value of the local tag at server $s$, at time $T$. Then, if $\pi$ denotes any completed repair or read operation that is initiated after time $T$, we have $tag(\pi) \geq \min_{s\in \mathcal{S}_m} t_s$. Also, if $\pi$ denotes any completed write  operation that is initiated after time $T$, we have $tag(\pi) > \min_{s\in \mathcal{S}_m} t_s$.
\end{lemma}
\begin{proof}
	Similar to the proof of Lemma \ref{lem:property}.		
\end{proof}	

Next, in order to apply Lemma \ref{lem:atom}, consider any well-formed execution $\beta$ of $\RADONS$, all of whose invoked read and write operations, denoted by the set $\Pi_{RW}$, complete.  Recall the discussion in Section \ref{sec:radons}, where we noted that tags for completed operations in $\RADONS$ are defined exactly as we had done for $\RADONL$. Thus, for any completed write operation $\pi$, we define $tag(\pi)$ as the tag created by the writer during the {\PutW} phase. For any completed read operation $\pi$, we define $tag(\pi)$ as the tag corresponding to the value returned by the read. Further, we define a partial order ($\prec$) on $\Pi_{RW}$ like in the proof of Theorem \ref{thm:atomicity} for case of $\RADONL$. These are restated for the sake of completion: For any $\pi, \phi \in \Pi_{RW}$, we say $\pi \prec \phi$  if one of the following holds: $(i)$  $tag(\pi)  < tag(\phi)$, or $(ii)$ $tag(\pi) = tag(\phi)$, and  $\pi$ and $\phi$ are write and read operations, respectively.

\subsection{Proof of Atomicity}

\paragraph{Property $P1$} Consider two successful operations $\phi$ and $\pi$ such  that $\phi$ completes before $\pi$ begins. We want to prove that $\pi \not\prec \phi$. Consider the case when both $\phi$ and $\pi$ are write operations (the other cases are similar, so only one case is discussed). By virtue of the definition of the partial order ($\prec$), it is enough to prove that $tag(\pi) > tag(\phi)$. Let $\mathcal{S}_{\alpha}$ and $\mathcal{S}_1$ respectively denote the set of servers whose responses were used by the writer during the {\PutData} and {\ConfirmData} phases of $\phi$. Let $T$ denote the time of initiation of the {\ConfirmData} phase of $\phi$. From the algorithm (see Fig. \ref{fig:radons-server}), we know that $\mathcal{S}_1 \subset \mathcal{S}_{\alpha}$. Further, based on the algorithm, it is clear that all servers in $\mathcal{S}_1$ (which is a majority) are active at time $T$, such that $t_{loc}(s)|_T \geq tag(\phi)$. In this case, we apply Lemma \ref{lem:property_radons} to conclude that $tag(\pi) > tag(\phi)$.

\paragraph{Property $P2$} This follows from the construction of tags, and the definition of the partial order ($\prec$).

\paragraph{Property $P3$} This follows from the definition of partial order ($\prec$), and by noting that value returned by a read operation $\pi$ is simply the value associated with $tag(\pi)$.

%%%%%%%%%%%%%%%%%%%%%%%%%%%%%%%%%%%%%%%%%%%5555

\end{document}